\documentclass[]{elsarticle}

\usepackage[T1]{fontenc}
\usepackage{stmaryrd} 
\usepackage{amsfonts}
\usepackage{amsmath}
\usepackage{amsthm}
\usepackage{xcolor}
\usepackage{upgreek}
\usepackage{bm}
\usepackage{tikz}
\usepackage{float}
\usepackage{caption}
\usepackage{subfig}
\usepackage{tabularx}
\usepackage{todonotes}
\usepackage{hyperref}
\usepackage[inline,shortlabels]{enumitem}
\usepackage[ruled,linesnumbered,vlined]{algorithm2e}
\newtheorem{thm}{Theorem}[section]
\newtheorem{definition}[thm]{Definition}
\newtheorem{remark}[thm]{Remark}
\newtheorem{example}[thm]{Example}
\newtheorem{lemma}[thm]{Lemma}
\newtheorem{proposition}[thm]{Proposition}

\usepackage[toc,page]{appendix}

\definecolor{input}{HTML}{303060}
\definecolor{output}{HTML}{804000}
\definecolor{string}{HTML}{A02020}
\definecolor{parent}{HTML}{A020A0}
\definecolor{function}{HTML}{205080}
\definecolor{constructor}{HTML}{205080}
\definecolor{method}{HTML}{205080}
\definecolor{keyword}{HTML}{008000}
\definecolor{error}{HTML}{B01010}
\definecolor{comment}{HTML}{60A060}
\definecolor{import}{HTML}{A46519}
\definecolor{vertex1}{RGB}{0,255,125}
\definecolor{vertex2}{RGB}{255,0,0}
\definecolor{vertex3}{RGB}{0,0,255}
\definecolor{vertex4}{RGB}{255,165,0}

\newcommand{\cIn}{{\color{input} \tt \phantom{C}In}:}
\newcommand{\cOut}{{\color{output} \tt Out}:}
\newcommand{\sage}{\textsc{SageMath}\xspace}

\newcommand{\R}[0]{\mathbb{R}}
\newcommand{\N}[0]{\mathbb{N}}
\newcommand{\Z}[0]{\mathbb{Z}}
\newcommand{\Q}[0]{\mathbb{Q}}

\newcommand{\lt}{\textnormal{lt}}

\newcommand{\lm}{\textnormal{lm}}

\newcommand{\lc}{\textnormal{lc}}

\newcommand{\valr}[1]{\textnormal{val}_{\mathbf{#1}}}

\newcommand{\val}{\textnormal{val}}

\newcommand{\rem}{\textnormal{rem}}

\newcommand{\valP}{\textnormal{val}_P}

\newcommand{\K}{K}
\newcommand{\Kz}{\K^\circ}

\newcommand{\X}{\mathbf{X}}

\newcommand{\XL}{\mathbf{X^{\pm 1}}}

\newcommand{\PolynomialMonoid}{T_{\ge 0}}

\newcommand{\LaurentMonoid}{T}

\newcommand{\PolynomialRing}{K[\X]}

\newcommand{\LaurentRing}{K[\XL]}

\newcommand{\LaurentPolytopalRing}{K\{\X;P\}}

\newcommand{\LaurentPolytopalRingForR}{K\{\X;r\}}

\newcommand{\MonoidOfTerms}{T\{\X\}}

\newcommand{\MonoidOfTermsForR}{T\{\X\}}

\newcommand{\MonoidOfTermsWithGMO}{T\{\X;\le_\omega\}}

\newcommand{\drawVector}[4]{
\pgfmathsetmacro{\vNorm}{sqrt((#1)^2 + (#2)^2)} % Norme du vecteur
\pgfmathsetmacro{\endX}{#3 * #1 / \vNorm}  % Coordonnée x finale
\pgfmathsetmacro{\endY}{#3 * #2 / \vNorm}  % Coordonnée y finale
\draw[#4] (0,0) -- (\endX,\endY);         % Dessin de la ligne
}

\begin{document}

\begin{frontmatter}
\title{Gröbner bases over polytopal affinoid algebras}

\author[1]{Moulay Barkatou}
\ead{moulay.barkatou@unilim.fr}
\author[1]{Lucas Legrand}
\ead{lucas.legrand@unilim.fr}
\author[1]{Tristan Vaccon}
\ead{tristan.vaccon@unilim.fr}

\affiliation[1]{organization={{Universit\'{e} de Limoges; CNRS, XLIM UMR 7252}},
city={Limoges},
country={France}}


\begin{abstract}
     Polyhedral affinoid algebras have been introduced by Einsiedler, Kapranov and Lind in \cite{EKL:2006}
     to connect rigid analytic geometry
     (analytic geometry over non-archimedean fields)
     and tropical geometry.
     In this article, 
     we present 
     a theory of Gröbner
     bases for 
     polytopal affinoid
     algebras that
     extends both
     Caruso et al.'s 
     theory of 
     Gröbner bases
     on Tate algebras of \cite{CVV:2019}
     and Pauer et al.'s 
     theory of Gröbner
     bases on Laurent
     polynomials of \cite{PU:1999}.
     
     We provide effective
     algorithms
     to compute Gröbner
     bases for both
     ideals of Laurent 
     polynomials and
     ideals in polytopal affinoid algebras.
     Experiments with
     a Sagemath implementation
     are provided.
\end{abstract}

\end{frontmatter}

%   <ccs2012>
%   <concept>
%   <concept_id>10010147.10010148.10010149.10010150</concept_id>
%   <concept_desc>Computing methodologies~Algebraic algorithms</concept_desc>
%   </concept>
%   </ccs2012>
% \ccsdesc[500]{Computing methodologies~Algebraic algorithms}
%
% \vspace{-1.5mm}
% \terms{Algorithms, Theory}
% \keywords{Algorithms, Gröbner bases, Tate algebra, Laurent polynomials,
% Tropical analytic geometry}
%
% \maketitle
%
% \widowpenalty = 10000
% \addtolength{\textfloatsep}{-0.45cm} % Distance between float (e.g. [t]) and text

\section{Introduction}

Rigid geometry was born in Tate's article \cite{Tate:1971}. He defined
affinoid algebras as ideal quotients of the Tate algebras,
the algebras of converging power series on the unit ball
of some complete non-archimedean valued field.
They are the building blocks of Tate's rigid geometry
in the same way ideal quotients in polynomial rings
are the building blocks of algebraic geometry.
In the past 50 years, they have found various applications,
among them one can cite Raynaud's contribution to solving
Abhyankar's conjecture \cite{Raynaud:1994}.

Polyhedral and polytopal affinoid algebras
have been defined by Einsiedler et al. in \cite{EKL:2006}
using Tate algebras with convergence conditions
given by a polyhedron or a polytope (respectively).
They are one of the main ingredients of the development
of tropical analytic geometry as in \textit{e.g.} \cite{Gubler:2007,Rabinoff:2012,FM:2023bis}.
It has found applications with Gubler in \cite{Gubler:2007bis}
to prove the Bogomolov conjecture for totally degenerate abelian varieties.
One motivation of this article is to make progress
toward an \textit{effective counterpart}
to tropical analytic geometry.

To do so, the most natural tool to implement
is Gröbner bases (GB).
The case of GB over Tate algebras has
been studied in \cite{CVV:2019}.
Modern algorithms like F5 and FGLM
have been generalized to this context
in \cite{CVV:2020,CVV:2021}.
Overconvergence has been studied in 
\cite{CVV:2022,VV:2023} culminating
in the definition and computation
of Universal Analytic Gröbner Bases.
They allow a first step toward
tropical analytic geometry 
in the context of polynomial ideals.

To extend these results to polytopal affinoid algebras,
it is natural to work with Laurent polynomials and series.
To define GB in this context,
one could of course, instead of working with $x_1^{\pm1},\dots,x_n^{\pm1}$,
choose to use nonnegative monomials $x_1,y_1,\dots,x_n,y_n$
along with relations $x_1y_1-1=0,\dots, x_ny_n-1=0$.
As we believe a direct approach would be more suitable for tropical applications,
we have chosen to implement an instance of Laurent polynomials
and series which does not hide its monomials with negative exponents,
and which has been introduced by Pauer and Unterkircher in \cite{PU:1999}.
Motivated by the study of systems of linear partial difference equations, 
they have developed a theory of generalized monomial orderings
and Gröbner bases for Laurent
polynomials.
We provide a short introduction to the generalized monomial ordering part of this theory in Section \ref{sec:GMO}. 
In addition, in Section \ref{sec:implem}, we fill in a gap that
prevented their theory to be completely effective: we
provide an algorithm to compute all the necessary leading monomials
in this context.

Equipped with a new notion of term ordering built on those of \cite{CVV:2019,VV:2023} and \cite{PU:1999}, 
we study the multivariate division in polytopal affinoid algebras
in Section \ref{sec:MultiVarDivision} in the particular case when the polytope is a point.
These results enable us to provide in Section \ref{sec:GBpolytopal}
a theory of Gröbner bases for ideals in polytopal affinoid domains in the particular case when the polytope is a point..
Adecisive step toward effectivity is made in Subsection  \ref{subsec:Buchberger}
with an adapted Buchberger algorithm.
Section  \ref{sec:implem}
provides additional tools needed for effective
GB computations.
Section \ref{sec:gb_valP} concludes the
article by extending the previous
results to the case of Gröbner bases for ideals in polytopal affinoid domains for 
a general polytope.

In addition, a short software demonstration in Sagemath \cite{sagemath}
can be read in Appendix.
Three proofs of results of
Section \ref{sec:gb_valP} are also given in the Appendix.

\par{Expanded version and shorter version:}
A shorter version of this article has been published in the Proceedings of the 49th International Symposium on Symbolic and Algebraic Computation (ISSAC 2024) \cite{BLV:2024}. The main differences between the shorter version and this paper are the following corrections and additions. 
Sections \ref{sec:MultiVarDivision} and \ref{sec:GBpolytopal}
have been revised and updated:
for their constructions,
a valuation compatible with
multiplication by 
monomials is needed and their
scope has been \emph{restricted} to
division and Gröbner bases in
$\LaurentPolytopalRingForR$ (see Section \ref{sec:setting}).
As such a completely new section
is added with Section \ref{sec:gb_valP} to define
and compute Gröbner bases for
ideals in $\LaurentPolytopalRing$ (see Section \ref{sec:setting}) using
a conic decomposition adapted
from the polytope $P$.

%\subsection{Notations}

% $n$ an integer $\ge 1$

% $\mathbf{X} = (X_1,\dots,X_n)$ and $\mathbf{X}^i = X_1^{i_1}\dots X_n^{i_n}$ for $i = (i_1,\dots,i_n) \in \Z^n$ 

% $M[\mathbf{X}^{\pm 1}]$ monoid of Laurent monomials.

% $F[\mathbf{X}^{\pm 1}]$ Ring of Laurent polynomial with coefficients in $F$

% $K$ non-archimedean valued field with discrete valuation (ex $Q_p$).

% $e_1,\dots,e_n =$ standard basis of $Z^n$.

\section{Setting} \label{sec:setting}
Let $K$ be a field with a discrete valuation $\val{}:K \to \mathbb{R}\sqcup \infty$
making it complete,
and let $\Kz$ be the subring of $K$ consisting of
elements of nonnegative valuation. 
Let $\pi$ be a uniformizer 
of $K$, $\textit{i.e.}$ a generator of the maximal ideal $\left\lbrace x \in \Kz, \: \val(x)>0 \right\rbrace$ of $\Kz.$
Typical examples of such a setting are $p$-adic fields like
 $K=\Q_p$ with $\Kz=\Z_p$ and
 $\pi = p$,
or Laurent series fields like $K=\Q(\!(U)\!)$ with $\Kz=\Q \llbracket U \rrbracket$ and
 $\pi = U$.

 We fix a positive integer $n$.
 Let $X_1,\dots,X_n$ be $n$ variables.
We use the short notations $\X$ for $(X_1,\dots,X_n)$ and $\X^{\pm 1}$ for $(X_1^{\pm 1},\dots,X_n^{\pm 1})$.
 If $\mathbf{i} = (i_1,\dots,i_n) \in \Z^n$, we shall write $\X^\mathbf{i}$ for $X_1^{i_1}\cdots X_n^{i_n}$.
 We define $\PolynomialMonoid$ and $\LaurentMonoid$ to be the multiplicative monoids $\{\X^\mathbf{i},\ \mathbf{i} \in \N^n\}$ and $\{\X^\mathbf{i},\ \mathbf{i} \in \Z^n\}$ respectively.
We will frequently represent a monomial in either the set $\PolynomialMonoid$ or $\LaurentMonoid$ by utilizing the $n$-tuple of its exponents in $\N^n$ or $\Z^n$.
%This association should be evident from the surrounding context and should not give rise to any confusion.

Let $\mathcal{P} \subset{\R^n}$ be a polytope with vertices $\textnormal{vert}(P) \subset \Q^n$. Let $P=\mathcal{P}\cap\mathbb{Q}^n$. 
For $a,b \in  \R^n$, $a \cdot b$ denote the usual scalar product in $\R^n$.
Following \cite{Rabinoff:2012}, we define the polytopal affinoid algebra $\LaurentPolytopalRing$ as the following algebra:
\[\left\{\sum_{u \in \Z^n}a_u\mathbf{X}^u\ :\ a_u \in K,\ \forall r \in P,\ \val(a_u) - r\cdot u \xrightarrow[|u| \rightarrow +\infty]{} +\infty \right\}.\]

% For $f = \sum_{u \in \Z^n}a_u \X^u \in \LaurentPolytopalRing$, we define $\textnormal{supp}(f) := \{ \X^u,\ u \in \Z^n,\ a_u \neq 0\}$.
% We define a valuation $\valr{r}$ on $\LaurentPolytopalRing$ by
% \[ \valr{r}(f) = \valr{r} \left(\sum_{u \in \Z^n}a_u\mathbf{X}^u \right) = \min_{u \in \Z^n}\val{}(a_u) - r\cdot u.\]
% The valuation $ \valr{r} $ is multiplicative: $\valr{r}(fg) = \valr{r}(f) + \valr{r}(g)$.
% We define $\valP(f) = \inf_{r \in P}\valr{r}(f)$.
% By \cite{Rabinoff:2012}, $\valP$ is again a valuation on $\LaurentPolytopalRing$ and 
% \[\valP(f) = \inf_{r \in P}\valr{r}(f) = \min_{r \in P}\valr{r}(f) = \min_{r \in \textnormal{vert}(P)}\valr{r}(f)\]
% Contrarly to $ \valr{r} $, $ \val_P $ is only sub-multiplicative: 
% $\val_{r}(fg) \ge \val_{r}(f) + \val_{r}(g)$.
% This difference between $ \valr{r} $ and $ \val_P $ wil play an important role in the paper.
% Moreover, $f$ is in $\LaurentPolytopalRing$ if and only if $\val(a_u) - r\cdot u \to +\infty$ for all $r \in \textnormal{vert}(P)$.
If $ P = \{r\}  $, we simply write $ K\{\X; r\} $ for $ K\{\X;\{r\} \} $.
From the definition of $ \LaurentPolytopalRing $, it is immediate that if $ P_1 \subseteq P_2 $, then 
$K \{\X; P_2\} \subseteq $ $ K \{\X;P_1\}  $.
For $f = \sum_{u \in \Z^n} a_u \X^u \in \LaurentPolytopalRing$, we define 
$\textnormal{supp}(f) := \{\X^u \mid u \in \Z^n, a_u \neq 0\}$.
For $ r \in P $, we define a valuation $\valr{r}$ on $\LaurentPolytopalRing$ as follows:
\[
\valr{r}(f) = \valr{r} \left(\sum_{u \in \Z^n} a_u \mathbf{X}^u \right) = \min_{u \in \Z^n} \val(a_u) - r \cdot u.
\]
The valuation $\valr{r}$ is multiplicative: $\valr{r}(fg) = \valr{r}(f) + \valr{r}(g)$.
We then define $\valP(f) = \inf_{r \in P} \valr{r}(f)$.
According to \cite{Rabinoff:2012}, $\valP$ is also a valuation on $\LaurentPolytopalRing$, and
\[
\valP(f) = \inf_{r \in P} \valr{r}(f) = \min_{r \in P} \valr{r}(f) = \min_{r \in \textnormal{vert}(P)} \valr{r}(f).
\]
Unlike $\valr{r}$, $\valP$ is only sub-multiplicative: 
$\valP(fg) \ge \valP(f) + \valP(g)$.
This distinction between $\valr{r}$ and $\valP$ will play an important role in the paper.
Additionally, $f$ is in $\LaurentPolytopalRing$ if and only if $\val(a_u) - r \cdot u \to +\infty$ for all $r \in \textnormal{vert}(P)$.

In other words:
\[ \LaurentPolytopalRing = \bigcap_{r \in  \textnormal{vert}(P)} K\{\mathbf{X};r\}\]
Elements of $\LaurentPolytopalRing$ are exactly the Laurent power series with coefficients in $K$ which converge on the set $\val^{-1}(P) \subset (\overline{K}^\times)^n$ (where $\textnormal{val}$ is extended to $\overline{K}$), and $\valP(f)$ is the minimum valuation reached by $f$ on $\val^{-1}(P)$.

A case of particular interest is when $P = \prod_{i=1}^n[r_i,s_i]$ for some $r_i < s_i \in \Q$.
For such $P$, $\val{}^{-1}(P)$ is the polyannulus 
$\{(\epsilon_1,\dots,\epsilon_n) \in (\overline{K}^\times)^n\ :\ r_i \le \val(\epsilon_i) \le s_i\}$ 
and in rigid geometry, such a $\LaurentPolytopalRing$ is called a Laurent domain.

\section{Generalized monomial order}
\label{sec:GMO}
In the realm of Gröbner basis theory for ideals within the polynomial ring $\PolynomialRing$, a fundamental component involves monomial orders defined on the monoid $\PolynomialMonoid$.
Recall that a monomial order on $\PolynomialMonoid$ is a total order that satisfies, for all $r, s, t \in \PolynomialMonoid$:

\begin{center}
  \begin{enumerate*}[(1),itemjoin={\hskip10mm}]
      \item $1 \le t$ 
      \item $r < s \implies rt < st$
  \end{enumerate*}
\end{center}

When attempting to extend this theory to Gröbner bases for ideals in the Laurent polynomial ring $\LaurentRing$, a natural inclination is to consider "monomial orders on $\LaurentMonoid$." 
These should be total orders on $\LaurentMonoid$ that adhere to conditions (1) and (2). 
However, such orders cannot exist.

To illustrate this, take any non-trivial element $t \in \LaurentMonoid$. 
According to (1), we have $1 < t$, which implies $t^{-1} < 1$ by (2). 
This leads to a contradiction. 
The problem is that, unlike $\PolynomialMonoid$, $\LaurentMonoid$ contains non-trivial invertible elements.

In the article \cite{PU:1999}, the authors introduced generalized monomial orders (or g.m.o) as a workaround to the previous problem.
Their approach involves representing $\LaurentMonoid$ as a finite union $\cup_i T_i$ of submonoids $T_i$, ensuring that each $T_i$ does not contain any non-trivial invertible elements, thereby allowing the definition of a monomial order on each of them.
The crucial requirement is that the order on $\LaurentMonoid$ must restrict to a monomial order on each $T_i$ in a compatible manner (refer to condition 2 in Definition \ref{def:gmo}).
This final condition ensures that the leading term of a product, which typically is not the product of the leading terms for a g.m.o, predominantly depends on the submonoids $T_i$.

In this section, we expose this notion, simplifying a little the exposition in \cite{PU:1999} to fit our use case.

\smallbreak
Sometimes, we will need to consider orders defined on various sets.
As a general guideline, the symbols $\le$ and $<$ will be employed irrespective of the set on which the order is established.
The clarity of the set should always be apparent from the context.
In instances where ambiguity might arise, we employ subscripts (such as $<_{G}$, $<_\omega$, $\dots$) to distinguish between different orders.

%\TVlong{Ce dernier paragrahe va peut etre dans la section notations, a voir}

\begin{definition}[{\cite[Definition 2.1]{PU:1999}}]
	\label{def:conic_decomposition}
	A conic decomposition of a sub-monoid $S \subseteq \LaurentMonoid$ is a finite family $(S_i)_{i \in I}$ of finitely generated submonoids of $S$ such that
	\begin{enumerate}[(1)]
		\item for each $i$, the only invertible element in the monoid $S_i$ is $1$ and the group generated by $S_i$ is $\LaurentMonoid$.
		\item the union of all the $S_i$'s equal $S$
	\end{enumerate}
\end{definition}

We will be exclusively interested in the case where $S = T$.
Examples \ref{ex:standard_conic_decomposition} and \ref{ex:second_conic_decomposition} illustrate Definition \ref{def:conic_decomposition}.
We will use the conic decomposition of Example \ref{ex:standard_conic_decomposition} in our implementation in \S \ref{subsec:particular_conic_decompo}
and in the Appendix.

\begin{example}
\label{ex:standard_conic_decomposition}
Let $T_0 := \{X^k,\ k \in \N^n\}$ and for $1 \le j \le n$ let $T_j$ be the monoid generated by $\{X_1^{-1}\dots X_n^{-1}\}\cup\{X_1,\dots,\hat{X_j},\dots,X_n\}$ where the hat symbol indicates that the corresponding element is omitted.
In other words, $T_j$ contains all monomials for which the exponent of $X_j$ is non-positive and smaller than any other exponent.
Then $(T_j)_{0 \le j \le n}$ is a conic decomposition of $\LaurentMonoid$ containing $n+1$ cones.
\end{example}

\begin{example}
\label{ex:second_conic_decomposition}
Let $D_n$ be the set of all maps from $\{1,\dots,n\}$ to $\{-1,1\}$. 
For $d$ in $D_n$, let $T_d$ be the monoid generated by the set $\{X_1^{d(1)},\dots,X_n^{d(n)}\}$.
Then $(T_d)_{d \in D_n}$ is a conic decomposition of $\LaurentMonoid$ containing $2^n$ cones.
\end{example}

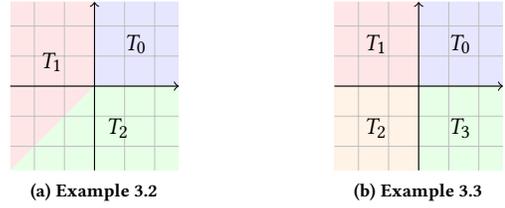
\begin{figure}
		\centering
		%\begin{subfigure}{0.45\textwidth}
		\subfloat[Example \ref{ex:standard_conic_decomposition}]{
		    \centering
			\begin{tikzpicture}[scale=0.8]
			\fill[blue!10!white] (0,0) rectangle (1.4,1.4);
			\fill[red!10!white] (0,0) rectangle (-1.4,1.4);
			\fill[red!10!white] (0,0) -- (-1.4,0) -- (-1.4,-1.4);
			\fill[green!10!white] (0,0) rectangle (1.4,-1.4);
			\fill[green!10!white] (0,0) -- (0,-1.4) -- (-1.4,-1.4);
			\draw[step=0.5cm,gray!50!white,very thin] (-1.4, -1.4) grid (1.4,1.4);
			\draw[thin,->] (-1.4,0) -- (1.4,0);
			\draw[thin,->] (0,-1.4) -- (0,1.4);
			\node at (0.7,0.7) {$T_0$};
			\node at (-0.7,0.4) {$T_1$};
			\node at (0.4,-0.7) {$T_2$};
			\end{tikzpicture}
			%\caption{Example \ref{ex:standard_conic_decomposition}}
			%\label{subfig1}
		} \quad \quad \quad \quad \quad \quad
		\subfloat[Example \ref{ex:second_conic_decomposition}]{
	        \centering
            \begin{tikzpicture}[scale=0.8]
			\fill[blue!10!white] (0,0) rectangle (1.4,1.4);
			\fill[red!10!white] (0,0) rectangle (-1.4,1.4);
			\fill[green!10!white] (0,0) rectangle (1.4,-1.4);
			\fill[orange!10!white] (0,0) rectangle (-1.4,-1.4);
			\draw[step=0.5cm,gray!50!white,very thin] (-1.4, -1.4) grid (1.4,1.4);
			\draw[thin,->] (-1.4,0) -- (1.4,0);
			\draw[thin,->] (0,-1.4) -- (0,1.4);
			\node at (0.7,0.7) {$T_0$};
			\node at (-0.7,0.7) {$T_1$};
			\node at (0.7,-0.7) {$T_3$};
			\node at (-0.7,-0.7) {$T_2$};
			\end{tikzpicture}
			%\caption{Example \ref{ex:second_conic_decomposition}}
			%\label{subfig2}
		}\caption{Conic decompositions for $n=2$}
	\label{schema:conic_decompositions}
\end{figure}

\begin{definition}[{\cite[Definition 2.2]{PU:1999}}]
\label{def:gmo}
	Let $(T_i)_{i \in I}$ be a conic decomposition of $\LaurentMonoid$.
	A generalized monomial order (or g.m.o) on $\LaurentMonoid$ for the decomposition $(T_i)_{i \in I}$ is a total order $<$ on $\LaurentMonoid$ such that
	\begin{enumerate}[(1)]
		\item $\forall t \in \LaurentMonoid, 1 \le t$
		\item $\forall r \in \LaurentMonoid, \forall i \in I$, $(s,t \in  T_i \ \textnormal{and}\ r < s) \implies rt < st$ \label{item:compatibilite_multiplication_dans_Ti}
	\end{enumerate}
\end{definition}

\begin{remark}
	For each $i \in I$, the restriction of $<$ to $T_i$ is a monomial order (take $r$ in $T_i$ in 2. of Definition \ref{def:gmo}).
\end{remark}

Given a conic decomposition, Lemma \ref{lemma:construct_gto} provides a method for constructing a g.m.o by employing an auxiliary function, $\phi: \LaurentMonoid \to \mathbb{Q}_{\ge 0}$, which exhibits favorable behavior in relation to the decomposition.

\begin{lemma}[{\cite[Lemma 2.1]{PU:1999}}]
\label{lemma:construct_gto}
Let $(T_i)_{i \in I}$ be a conic decomposition of $\LaurentMonoid$ and $E$ be either $\{1\}$ or one of the $T_i$.
Let $<_G$ be a total group order on $\LaurentMonoid$ (e.g. the lexicographical order). 
Let $\phi : \LaurentMonoid \to \mathbb{Q}_{\ge 0}$ be a function fulfilling the following conditions:
\begin{enumerate}[(1)]
	\item $\forall t \in \LaurentMonoid\setminus E$, $\phi(t) > 0$
	\item $\forall s,t \in \LaurentMonoid$, $\phi(st) \le \phi(s) + \phi(t)$
	\item $\forall i \in I$, $\phi|_{T_i}$ is a monoid homomorphism.
\end{enumerate}
\noindent Then the order $<$ defined by
\[ r < s \iff \phi(r) < \phi(s) \textnormal{ or } (\phi(r) = \phi(s) \textnormal{ and } r <_G s)\]
is a g.m.o on $\LaurentMonoid$ for the decomposition $(T_i)_{i \in I}$.
\end{lemma}

The notions of leading monomial, leading coefficient and leading term are defined for a Laurent polynomial as in the polynomial case.
\begin{definition}
\label{def:leadings}
Fix a g.m.o on $T$ and let $f \in \LaurentRing$. 
The leading monomial $\lm(f)$ of $f$ is defined as $\max(\X^j,\ j \in \textnormal{supp(f)})$.
The leading coefficient $\lc(f)$ of $f$ is the coefficient of $\lm(f)$ in $f$. 
The leading term $\lt(f)$ of $f$ is the product $\lc(f)\lm(f)$.
\end{definition}

Examples \ref{ex:standard_gmo_1} to \ref{ex:second_gmo} illustrate Lemma \ref{lemma:construct_gto} and Definition \ref{def:leadings}.
In each case, the group order on $\LaurentMonoid$ is taken to be the lexicographical order, and we show in the case $n=2$ how the monomials of $f = 2xy^{-2} + x^{-2}y^{-2} + 3x^{-1}y^{-2} + y^2 \in K[x^{\pm 1},y^{\pm 1}]$ are ordered.

\begin{example}
\label{ex:standard_gmo_1}
Take the standard decomposition $(T_0,T_1,\dots,T_n)$ of Example \ref{ex:standard_conic_decomposition}.
Let $E = \{1\}$ and define $\phi : \LaurentMonoid \to \mathbb{Q}_{\ge 0}$ by $\phi(i_1,\dots,i_n) = i_1 + \dots + i_n - (n+1)\min(0,i_1,\dots,i_n)$.
We have $xy^{-2} > x^{-1}y^{-2} > y^2 > x^{-2}y^{-2}$, $\lm(f)= xy^{-2}$, $\lc(f) = 2$ and $\lt(f) = 2xy^{-2}$.
\end{example}

\begin{example}
\label{ex:standard_gmo_2}
Take the standard decomposition $(T_0,T_1,\dots,T_n)$ of Example \ref{ex:standard_conic_decomposition}. 
Let $E = T_0$ and define $\phi : \LaurentMonoid \to \mathbb{Q}_{\ge 0}$ by $\phi(i_1,\dots,i_n) = -\min(0,i_1,\dots,i_n)$. 
We have $xy^{-2} > x^{-1}y^{-2} > x^{-2}y^{-2} > y^2$, $\lm(f) = xy^{-2}$, $\lc(f) = 2$ and $\lt(f) = 2xy^{-2}$.
\end{example}

\begin{example}
\label{ex:second_gmo}
Take the conic decomposition $(T_d)_{d \in D_n}$ of Example \ref{ex:second_conic_decomposition}. 
Let $E = \{1\}$  and define $\phi: \LaurentMonoid	\to \mathbb{Q}_{\ge 0}$ by $\phi(i_1,\dots,i_n) = \lvert i_1 \rvert + \dots + \lvert i_n \rvert $. 
We have $x^{-2}y^{-2} > xy^{-2} > x^{-1}y^{-2} > y^2$, $\lm(f) = x^{-2}y^{-2}$, $\lc(f) = 3$ and $\lt(f) = 3x^{-2}y^{-2}$.
\end{example}

In Example \ref{ex:problem_leading_term}, we illustrate the fact that for $t \in \LaurentMonoid$ and $f \in \LaurentRing$, the leading monomial of $tf$ is generally not equal to the product of $t$ by $\lm(f)$.

\begin{example}
\label{ex:problem_leading_term}
Take the g.m.o of Example \ref{ex:standard_conic_decomposition} in the case $n=2$ and $f = xy + y^{-1} \in K[x^{\pm 1}, y^{\pm 1}]$.
We have $\lm(yf) = xy^2$ but $y\lm(f) = 1 \neq xy^2$.
\end{example}

However, thanks to the compatibility condition (2) of Definition \ref{def:gmo}, if two monomials $t_1$, $t_2$ are such that $\lm(t_1f)$ and $\lm(t_2f)$ lies in the same cone $T_i$, then the monomials $a_1, a_2$ of $f$ such that $\lm(t_1f) = t_1a_1$ and $\lm(t_2f) = t_2a_2$ are equal. 
This is proven in Lemma \ref{lemma:lmi_independent} and leads naturally to the definition of "one leading monomial per cone" in Definition \ref{def:lmi}.
For the rest of this section, we fix a g.m.o for a conic decomposition $(T_i)_{i \in I}$ of $T$.
\begin{definition}
	\label{def:T_i(f)}
	For $i \in  I$ and $f \in \LaurentRing$, define \[T_i(f) := \{ t \in \LaurentMonoid,\ \lm(tf) \in T_i\}\].
\end{definition}

\begin{lemma}[{\cite[Lemma 2.3]{PU:1999}}]
	\label{lemma:lmi_independent}
	Let $i \in I$, $f \in \LaurentRing$, $u,v \in T_i(f)$. 
	Write $\lm(uf) = ut_u \in T_i$ and $\lm(vf) = vt_v \in T_i$ for some monomials $t_u,t_v$ of $f$. 
	Then $t_u = t_v$. 
\end{lemma}

\begin{definition}
	\label{def:lmi}
	Let $i \in  I$, $f \in  \LaurentRing$ and $t \in T_i(f)$. 
	We define $\lm_i(f) := \lm(tf)t^{-1}$. 
	This is well defined by Lemma \ref{lemma:lmi_independent} (i.e it does not depend on a particular $t \in T_i(f)$).
	We also define $\lc_i(f)$ as the coefficient of $\lm_i(f)$ in $f$, and $\lt_i(f) = \lc_i(f)\lm_i(f)$.
\end{definition}

\begin{remark}
   Before giving some examples, we explain how to compute $\lm_i(f)$. 
By Lemma \ref{lemma:lmi_independent}, we have $\lm_i(f) = \lm(tf)t^{-1}$ whenever $t \in T_i(f)$. 
So we just need to find a $t \in T_i(f)$, which can be done as follows.
Recall that $T_i$ is finitely generated and generates $T$ as a group. 
This implies that for each monomial $s$ of $f$, we can find $u_s, v_s \in T_i$ such that $s = u_sv_s^{-1}$. 
Define $t$ as the product of the monomials $v_s$ for $s$ a monomial of $f$.
Then $\textnormal{supp}(tf) \subset  T_i$, and $\lm(tf) \in T_i$. 
The latter means that $t \in T_i(f)$ and we are done. \label{rem:how_to_compute_the_Tis}
\end{remark}

%Examples \ref{ex:standard_gmo_1_2} to \ref{ex:second_gmo_2}  continue Examples \ref{ex:standard_gmo_1} to \ref{ex:standard_gmo_2}.

\begin{example}[Example \ref{ex:standard_gmo_1} continued.]
\label{ex:standard_gmo_1_2}
 We have $\lm_0(f) = \lm_1(f) = y^2$, $\lm_2(f) = \lm(f) = xy^{-2}$, $T_0(f) = y^2T_0$, $T_1(f) = y^2T_1$ and $T_2(f) = yT_2$.
\end{example}

\begin{example}[Example \ref{ex:standard_gmo_2} continued.]
\label{ex:standard_gmo_2_2}
 We have $\lm_0(f) = \lm_2(f) = \lm(f) = xy^{-2}$, $\lm_1(f) = x^{-2}y^{-2}$, $T_0(f) = x^2y^2T_0$, $T_1(f) = xy^2T_1$ and $T_2(f) = x^2y^2T_2$.
\end{example}

\begin{example}[Example \ref{ex:second_gmo} continued.]
\label{ex:second_gmo_2}
We index the 4 cones as in b) of Figure \ref{schema:conic_decompositions}.
We have $\lm_0(f) = \lm_1(f)= y^2$, $\lm_2(f) = \lm(f) = x^{-2}y^{-2}$ and $\lm_3(f)  = xy^{-2}$.
\end{example}

\begin{remark}
	We have $\lm(f)$ = $\lm_i(f)$ for at least one index $i \in I$. 
	For a fixed $f$, the $\lm_i(f)$'s are not necessarily distincts, and $\lm_i(f)$ is generally not an element of $T_i$.
\end{remark}

The following lemma states that a g.m.o shares with monomial orders the crucial property of being a well-order.
\begin{lemma}[{\cite[Lemma 2.2]{PU:1999}}]
\label{lemma:stric_descending_seq}
Let $<$ be a g.m.o on $\LaurentMonoid$ (for a conic decomposition $(T_i)_{i \in I}$ of $\LaurentMonoid$). 
Then every strictly descending sequence in $\LaurentMonoid$ is finite. 
In particular, any subset of $\LaurentMonoid$ contains a smallest element.
\end{lemma}

\section{Multivariate division in $K\{\X;r\} $}
\label{sec:MultiVarDivision}
In this section, we demonstrate the adaptability of the division algorithm, originally presented in the Tate algebras setting in \cite[Proposition 3.1]{CVV:2019}, to the polytopal setting in the case $ P = \{r\}  $ is a single point. The general case will be studied in Section \ref{sec:gb_valP}.

The novelty lies in the method employed to eliminate the leading term of an element $f \in \LaurentPolytopalRingForR$ by $g \in  \LaurentPolytopalRingForR$.
We can not solely rely on the leading terms of $f$ and $g$ (see Example \ref{ex:division_failure}), but have to consider also the terms $\lt_i(g)$ and the sets $T_i(g)$ for $i \in I$.

\subsection{The monoid of terms $ \MonoidOfTermsForR $}
\label{subsec:monoid_of_terms}

\begin{definition}
	\label{def:order}
Let $\le_\omega$ be a g.m.o on $\LaurentMonoid$. 
We define the monoid of terms $\MonoidOfTermsWithGMO$ (or simply $\MonoidOfTerms$) as the multiplicative monoid $\{a\mathbf{X}^u:\ a \in K^\times, u \in \Z^n\}$.
We define a preorder $\le_{r}$ on $\MonoidOfTerms$ by
\begin{align*}
	a\X^u \le_r b\X^v \iff &(\valr{r}(a\X^u) > \valr{r}(b\X^v)) \textnormal{ or }\\ &(\valr{r}(a\X^u) = \valr{r}(b\X^v) \textnormal{ and } \X^u \le_{\omega} \X^v) 
\end{align*}
 \end{definition}

\begin{remark}
The preorder $\le_r$ is not antisymmetric (and so not an order).
For terms $t_1,t_2$, the fact that $t_1 \le_r t_2$ and $t_2 \le_r t_1$ is equivalent to the existence of $a \in (\Kz)^\times$ such that $t_1 = at_2$.
\end{remark}
\begin{remark}
The preorder $\le_r$ is compatible with multiplication by elements of $K^\times$ but is not compatible with multiplication
by monomials (because $\le_\omega$ itself is not).
It is also not a well-order (there exists infinite strictly decreasing sequences). 
\end{remark}

However, Lemma \ref{lemma:topological_order} shows that the preorder $\le_r$ is a \textit{topological} well-order as in \cite[Lemma 2.14]{CVV:2019}.

\begin{lemma}
	\label{lemma:topological_order}
	Let $(t_j)_{j \in \N}$ be a strictly decreasing sequence for $ \le_r $ in $\MonoidOfTerms$. 
	Then $\lim_{j \to \infty}\valr{r}(t_j) = +\infty$, or equivalently $\sum_{j \in \N}t_j \in\LaurentPolytopalRingForR\}$ .
\end{lemma}

\begin{proof}
	By definition of the preorder $\le_r$, the sequence $(\valr{r}(t_j))_{j \in \N}$ 
is nondecreasing and takes its values in a discret subset $\frac{1}{D}\Z$ of $\R$ 
for some integer $D$ (taking into account the image group of $\val$ and 
the product of the denominators of $ r $). 
Since $\le_\omega$ is a well-order, there can be for a fixed 
$v \in \frac{1}{D}\Z$ only a finite number of indices $j$ for which $\valr{r}(t_j) = v$. 
Combining these two facts, $\valr{r}(t_j)$ must tend to $+\infty$.
\end{proof}

\begin{remark}
	\label{remark:leading_term}
	If $i \neq j \in \Z^n$, the terms $a_i\X^i$ and $b_j\X^j$  are never "equal"
	(that is $a_i\X^i \le_r b_j\X^j$ and $b_j\X^j \le_r a_i\X^i$) for $\le_r$. 
	Therefore, any nonzero $f \in \LaurentPolytopalRingForR$ has a unique leading term for $\le_r$.
\end{remark}

\begin{definition}
Let $f=\sum_{\alpha \in \Z^n} c_\alpha \X^\alpha \in \LaurentPolytopalRingForR$.
We define
\begin{equation}
	\textnormal{in}_r(f) := \sum_{\alpha \in \Z^n \textrm{ s.t. } \valr{r}(c_\alpha \X^\alpha)= \valr{r}(f)}  c_\alpha \X^\alpha  \in \LaurentRing.
\end{equation}

Using the map $\textnormal{in}_r: \LaurentPolytopalRingForR \to \LaurentRing$,
we extend to $f \in \LaurentPolytopalRingForR$ the definitions of 
$\lm(f)$, $\lc(f)$, $\lt(f)$, $\lm_i(f)$, $\lc_i(f)$, $\lt_i(f)$ and $T_i(f)$. 
\end{definition}

The leading term for the term order on $\MonoidOfTerms$ (as in Remark \ref{remark:leading_term})
equals $\lt(f) = \lt(\textnormal{in}_r(f))$, and same for $\lt_i(f), \, T_i(f)$ and so on.

\subsection{Multivariate division algorithm}
\label{subsec:division}

In Example \ref{ex:division_failure}, we illustrate how the classical procedure to cancel
the leading term of a polynomial can fail in the Laurent polynomial setting equipped with a g.m.o.
\begin{example}
	\label{ex:division_failure}
	Take the g.m.o of Example \ref{ex:standard_gmo_1} in the case $n=2$ and $ f = x + y$, $g = x^{-1}y + y^{-1} \in K[x^{\pm 1},y^{\pm1}]$.
	We have $\lt(f) = x$ and $\lt(g) = x^{-1}y$. 
	To cancel out $\lt(f)$ in $f$ the classical way, we substract $\frac{\lt(f)}{\lt(g)}g$ from $f$ and obtain $y - x^2y^{-2}$.
	But then we have $\lm(y - x^2y^{-2} ) = x^2y^{-2} > x = \lm(f)$, that is the leading monomial after cancellation is strictly superior to the leading monomial of $f$!
\end{example}

The issue in Example \ref{ex:division_failure} arises from assuming that $\lt(tg) = t\lt(g)$.
For a g.m.o, this equality does not necessarily hold (e.g Example \ref{ex:problem_leading_term}),
and it becomes possible for the leading term of $tg$ to exceed the leading term of $f$. 
While the leading term of $f$ is successfully canceled out, a larger term originating from $tg$ emerges after the cancellation.
To properly eliminate the leading term of $f$ by a multiple of $g$, we need instead to identify a term $t$ such that $\lt(tg) = \lt(f)$. 
According to Lemma \ref{lemma:lmi_independent}, depending on the specific cone $T_i$ in
which $\lm(f)$ is contained, this equality can be expressed again as $\lt(f) = \lt(tg) = t\lt_i(g)$.
This shows that the potential candidates for $t$ are the terms $\frac{\lt(f)}{\lt_i(g)}$
for varying indices $i$ for which $\lm(f) \in T_i$.
Summarizing the process:
\begin{enumerate}[(1)]
	\item Find an index $i$ such that $\lm(f) \in T_i$ (there exists at least one).
	\item Check if $\lm \left(\frac{\lm(f)}{\lm_i(g)}g\right) = \lm(f)$. 
		If this condition is satisfied, the leading term of $f$ can be effectively canceled out by subtracting $\frac{\lt(f)}{\lt_i(g)}g$ from $f$, and $\frac{\lm(f)}{\lm_i(g)} \in  T_i(g)$; otherwise, it cannot.
\end{enumerate}

\begin{remark}
	The monomial $\lm(f)$ can be contained in more than one cone, but it suffices to test the condition for cancellation within any one of those cones.
\end{remark}

The preceding discussion is applied in Proposition \ref{prop:multi_div} and Algorithm \ref{alg:multi_div}
to formulate a multivariate division algorithm in the ring $\LaurentPolytopalRingForR$ using $ \le_r $ to order terms.
Same as \cite[Algo 1]{CVV:2019}, Algorithm \ref{alg:multi_div} needs a countable amount of steps to terminate,
but only a finite amount of steps to reach  a given finite precision in $\valr{r}.$
\begin{proposition}
	\label{prop:multi_div}
	Let $f \in \LaurentPolytopalRingForR$ and $G$ be a finite subset of $\LaurentPolytopalRingForR$. 
	Algorithm \ref{alg:multi_div} produces a family $(q_g)_{g \in  G}$ and $s$ in $\LaurentPolytopalRingForR$ such that:
	\begin{enumerate}[(1)]
		\item $f = \sum_{g \in  G}q_gg + s$
		\item for all monomial $t$ in $s$, $t \notin \bigcup_{i \in I, g \in G}T_i(g)\lm_i(g)$
		\item for all $g \in  G$ and all monomial $t$ in $q_g$, $\lt(tg) \le_r \lt(f)$.
	\end{enumerate}
\end{proposition}

\begin{proof}
We construct by induction sequences $(f_j)_{j\ge 0}$, $(q_{g,j})_{j \ge 0}$ for $g \in G$ and 
$(s_j)_{j\ge 0}$ such that for all $j \ge 0$:
\[ f = f_j + \sum_{g \in G}q_{g,j}g + s_j,\]
and $\lt(f_j)_{j\ge 0}$ is strictly decreasing for $ \le_r$.
We first set $f_0 = f$, $s_0 = 0$ and $q_{g,0} = 0$ for all $g \in G$.
If there exists $i \in I$ and $g \in G$ such that
\[\lm \left( \frac{\lm(f_j)}{\lm_i(g)}g \right) = \lm(f_j) \in T_i, \]
we set $f_{j+1} = f_j - tg$ and $q_{g,j+1} = q_{g,j} + tg$ where $t = \frac{\lt(f_j)}{\lt_i(g)}$,
and leave unchanged $s_j$ and the other $q_{g,j}$'s. 
Otherwise, we set $f_{j+1} = f_j - lt(f_j)$ and $s_{j+1} = s_j + lt(f_j)$ and leave unchanged the $q_{g,j}$'s. 
By construction, the sequence $(\lt(f_j))_{j \ge 0}$ is strictly decreasing for $ \le_r $.
By Lemma \ref{lemma:topological_order}, we deduce that 
$\valr{r}(s_{j+1}-s_j)$ and the $\valr{r}(q_{g,j+1}-q_{g,j})$'s tend to $+\infty$ when $j \to +\infty$.
Thus $s_j$ and the $q_{g,j}$'s converge in $K\{\mathbf{X};\{r\} \}$.
Their limits % elements $q_g=\lim_{j \to +\infty}q_{g,j}$ for $g \in G$ and $r = \lim_{j \to +\infty}r_j$ 
satisfy the requirements of the proposition.
\end{proof}

\begin{definition}
    The output $s$ obtained from Algo \ref{alg:multi_div} with entries $f$ and $G=(g_1,\dots,g_m)$ is denoted $\rem(f,G).$
\end{definition}

\begin{algorithm}
	\SetKwInOut{Input}{input}\SetKwInOut{Output}{output}
	\caption{Multivariate division algorithm in $\LaurentPolytopalRingForR$ for $ \le_r $}
	\label{alg:multi_div}
	\Input{$f,g_1,\dots,g_m \in \LaurentPolytopalRingForR$}
	\Output{$q_1,\dots,q_m,s \in \LaurentPolytopalRingForR$ satisfying Prop \ref{prop:multi_div}}
	\BlankLine
	$q_1,\dots,q_m,s\gets 0$\;
	\While{ $f \neq 0$}{
		%$i \gets$ index in $I$ such that $\textnormal{lm}(f) \in T_i$\;
		\While{ $\exists (i,j) \in I \times \llbracket 1,m \rrbracket$ such that $\lm\left(\frac{\lm(f)}{\lm_i(g_j)}g_j\right) = \lm(f)$}{
			$t \gets \frac{\lt(f)}{\lt_i(g_j)}$\;
			$q_j \gets q_j + t $\;
			$f \gets f - tg_j $\;
		}
		$s \gets s + \lt(f)$

		$f \gets f - \lt(f)$;
	}
	\Return $q_1,\dots,q_m,s$
\end{algorithm}

\section{Gröbner bases for ideals in $ \LaurentPolytopalRingForR $}
\label{sec:GBpolytopal}

Let $\le$ be a g.m.o for a conic decomposition $(T_i)_{i \in I}$. 
For an ideal $J$ in $\LaurentPolytopalRingForR$, define $\lm(J) := \left\{ \lm(f),\ f \in J,\ f \neq 0 \right\}$.
Notice that 

\begin{equation}
	\label{eq:lmJ}
\lm(J) = \bigcup_{f \in  J, i \in I}T_i(f)\lm_i(f).
\end{equation}

A finite subset $G$ of $J$ is then a Gröbner basis when equation \eqref{eq:lmJ} still holds while considering in the set union only elements from $G$ instead of all members of $J.$ 

\subsection{Gröbner bases}

\begin{definition}
	Let $J$ be an ideal in $\LaurentPolytopalRingForR$ and $G$ be a finite subset of $J\setminus\{0\}$. 
	We say that $G$ is a Gröbner basis of $J$ (with respect to the g.m.o $\le$ and the conic decomposition $(T_i)_{i \in I}$) when:
	\[ \lm(J) = \bigcup_{g \in G,i \in I}T_i(g)\lm_i(g).\]
\end{definition}

In Proposition \ref{prop:has_gb}, we prove that every ideal $J$ within $\LaurentPolytopalRingForR$ contains a Gröbner basis. 
The demonstration relies on the fact that the set $T_i(f)$ is a finitely generated $T_i$-module, as proved in the following lemma:

\begin{lemma}
\label{lemma:fin_gen}
For all $i \in I$, there is a finite subset $E_i \subseteq \lm(J)\cap T_i$ such that $\lm(J)\cap T_i = T_i\cdot E_i$.
In case $J = \langle f \rangle$ is principal, we have $\lm(\langle f \rangle)\cap T_i = T_i(f)\lm_i(f) = T_i\cdot E_i$, and so $T_i(f) = T_i\cdot  \Bigl\{\frac{t}{\lm_i(f)},\ t \in E_i\Bigr\}$ is finitely generated over $T_i.$
\end{lemma}
\begin{proof}
	By Dickson's lemma, there exists a finite subset $E_i$ of $\lm(J)\cap T_i$ which generates the ideal $\langle \lm(J)\cap T_i \rangle_{K[T_i]}$ of $K[T_i]$.
	Then $\lm(J)\cap T_i =  T_i\cdot E_i$. 
	If $(J) = \langle f \rangle$ is principal, then $\lm(\langle f \rangle)\cap T_i = T_i(f)\lm_i(f)$ and we can conclude with Def. \ref{def:lmi}.
\end{proof}

\begin{proposition}
\label{prop:has_gb}
Let $J$ be an ideal of $\LaurentPolytopalRingForR$. Then $J$ contains a Gröbner basis.
\end{proposition}

\begin{proof}
 By Lemma \ref{lemma:fin_gen}, there are finite subsets $E_i \subset T_i$ such that $\lm(J) = \bigcup_{i \in I}T_i\cdot E_i$. 
 For all $t \in \bigcup_{i \in I}E_i$ choose an element $f_t \in J$ such that $\lm(f_t) = t$. Thanks to Def \ref{def:lmi}, $\lm_i(f_t)=\lm(f_t)=t.$
 We prove that the finite set $G=\{f_t \ |\ t \in \bigcup_{i \in I}E_i\}$ is a Gröbner basis of $J$.
 Let $u \in \lm(J).$ Let $i$ be such that $u \in T_i.$ Then there is some $(v,t) \in T_i \times E_i$ such that 
 $u=vt.$
 Let $c_t = \lc(f_t)$ and $c_\alpha \X^\alpha$ be a term of $f_t.$
 Then $c_\alpha \X^\alpha <_r c_t t$. 
 We remark that on the one hand $\valr{r} (c_\alpha \X^\alpha ) > \valr{r} (c_t t)$ 
 if and only if $ \valr{r} (c_\alpha v \X^\alpha ) > \valr{r} (c_t vt)$,
 and on the other hand, if $\valr{a} (c_\alpha \X^\alpha ) = \valr{r} (c_t t)$ ,
 then by item \ref{item:compatibilite_multiplication_dans_Ti} of Def \ref{def:gmo}
 and the fact that $v$ and $t$ are in $T_i$, we can remark that
 $\X^\alpha <_\omega t$ implies $\X^\alpha v <_\omega tv.$
 In any case, $c_\alpha \X^\alpha v <_r c_t vt$.
 Thus, $\lm(v f_t)=vt=u \in T_i.$
 Consequently, $v \in T_i(f_t)$ and $u \in T_i(f_t) \lm_i(f_t).$
We can then conclude that $G$ is a GB of $J.$
\end{proof}

\begin{proposition}
\label{prop:classical_easy}
Let $G$ be a Gröbner basis for an ideal $J$ of $\LaurentPolytopalRingForR$. We have
\begin{enumerate}[(1)]
	\item   $\forall f \in \LaurentPolytopalRingForR$, $f \in J \iff \rem(f,G) = 0$
	\item  $G$ generates the ideal $J$
\end{enumerate}
\end{proposition}

\begin{proof}
For the first part, if $\rem(f,G) = 0$, then the multivariate division algorithm gives $f \in J$. 
	Reciprocally, if $f \in J$ and $G$ is a Gröbner basis of $J$, then by (2) of Prop \ref{prop:multi_div} the remainder of the division of $f$ by $J$ is necessarily $0$.
For the second part, if $f \in J$,then, thanks to the above first part, $\rem(f,G) = 0$, and so $f \in J$.
\end{proof} 

\subsection{S-pairs and Buchberger criterion}

In the polynomial case, the S-pair of two polynomials $f$ and $g$ is formed to eliminate the minimal multiple of their leading terms, relying on the concept of least common multiple in $\PolynomialMonoid$.

However, in the Laurent setting, where $\LaurentMonoid$ is a group and the notions of lcm and gcd of two monomials in $T$ are useless, the replacement for the lcm of $f$ and $g$ is the finite set introduced in Proposition \ref{prop:Uifg}.
Notice that this set now depends on the cone in which the cancellation occurs.

By the same argument as in Lemma \ref{lemma:fin_gen}, we have:
\begin{proposition}
	\label{prop:Uifg}
	For $i \in I$ and $f,g \in \LaurentPolytopalRingForR$,
	$\lm_i(f)T_i(f) \cap \lm_i(g)T_i(g) \subset T_i$ is a finitely generated $T_i$-module.
\end{proposition}

\begin{definition}[S-pair]
	\label{def:Spair}
	Let $f,g \in \LaurentPolytopalRingForR$ and $i \in I$. 
	For $v \in \lm_i(f)T_i(f) \cap \lm_i(g)T_i(g)$, we define:
	\[ S(i,f,g,v) := \lc_i(g)\frac{v}{\lm_i(f)}f - \lc_i(f)\frac{v}{\lm_i(g)}g.\]
\end{definition}

\begin{lemma}
	\label{lemma:Spair}
	For $f,g \in \LaurentPolytopalRingForR$, $i \in I$ and $v \in \lm_i(f)T_i(f) \cap \lm_i(g)T_i(g)$, we have
    $\lt(S(i,f,g,v)) <_r \lc_i(f)\lc_i(g)v$.
\end{lemma}
\begin{proof}
Since $v \in T_i(f)\lm_i(f)\cap T_i(g)\lm_i(g)$, there exists $m_f \in T_i(f)$ and
$m_g \in T_i(g)$ such that $v = \lm(m_ff) = \lm(m_gg) = m_f\lm_i(f) = m_g\lm_i(g)$. 
Then the leading terms of $\lc_i(g)\frac{v}{\lm_i(f)}f$ and  $\lc_i(f)\frac{v}{\lm_i(g)}g$
both equal $\lc_i(f)\lc_i(g)v$. They cancel out leaving $\lt(S(i,f,g,v)) <_r \lc_i(f)\lc_i(g)v$.
\end{proof}

Since it involves all objects introduced so far, we give a detailed proof of the following adaptation of the classical cancellation lemma (see also \textit{e.g.} \cite[Lemma 5.1]{CVV:2022}).
\begin{lemma}
	\label{lemma:sumSpair}
Let $h_1,\dots,h_m \in \LaurentPolytopalRingForR$ and $i \in I$. 
For $1\le j \le m-1$, let $U(i,h_j,h_{j+1})$ be a finite system of generators of  $\lm_i(h_j)T_i(h_j) \cap \lm_i(h_{j+1})T_i(h_{j+1})$ which exists by Proposition \ref{prop:Uifg}.
Suppose that there are $t_1,\dots, t_m \in \MonoidOfTerms$, $u \in T_i$ and $c \in \val{}(K^\times)$ such that
\begin{itemize}
	\item for all $j \in \{1,\dots,m\}$, $\lt(t_jh_j) = c_ju$ with $\textnormal{val}(c_j) = c$
	\item  $\lt(\sum_{i=1}^{m}t_jh_j) <_r c_1u$.
\end{itemize}

Then there are elements $d_j \in K$, $v_j \in U(i,h_j,h_{j+1})$ for $1 \le j \le m-1$ and $t_m^\prime \in \LaurentPolytopalRingForR$ such that:
\begin{enumerate}[(1)]
	\item $\sum_{j=1}^{m}t_jh_j = \sum_{j=1}^{m-1}d_j\frac{u}{v_j}S(i,h_j,h_{j+1},v_j) + t_{m}^\prime h_m$. \label{enum:1}
	\item $\valr{r}(t^\prime_mh_m) > \valr{r}(uc_1)$. \label{enum:2}
	\item $\frac{u}{v_j} \in T_i$ for all $j <m$. \label{enum:3}
	\item For all $j <m$, $\val \left( d_j \lc_i(h_j) \lc_i(h_{j+1}) \right) \geq c.$ \label{enum:4}
	% \item $\val{P}(d_j\frac{u}{r_j}) + \max(\val{P}(h_j),\val{P}(h_{j+1})) \ge \val{P}(uc_1)$ for $j < m$.
\end{enumerate}

\end{lemma}
\begin{proof}
	Write $p_j = \frac{t_jh_j}{c_j}$, $e_j = \sum_{k=1}^{j}c_k$ and $t_j=\gamma_j \tilde{t}_j$ for some $\gamma_j \in K$
	and some monomial $\tilde{t}_j$.
	By hypothesis $u$ is in $T_i$ and $u = \tilde{t_j}\lm_i(h_j) \in T_i(h_j)\lm_i(h_j)$ for all $j$. 
	This implies that for all $j <m$ we have:
	\[u \in  T_i(h_j)\lm_i(h_j) \cap T_i(h_{j+1})\lm_i(h_{j+1}) = T_i\cdot U(i,h_j,h_{j+1})\]
	We deduce that for all $j<m$, there exist $k_j \in T_i$ and $v_j \in  U(i,h_j, h_{j+1}$) such that $u = k_jv_j$.
	Now write
	\begin{align}
	\sum_{j=1}^{m}t_jh_j &= e_1(p_1 - p_2) + \dots + e_{m-1}(p_{m-1}- p_m) + e_mp_m \label{eqn:somme_telescopique}
	\end{align}
	For all $j<m$, we have $\lt(t_j h_j)=c_j u= \gamma_j \lc_i (h_j)\tilde{t}_j \lm_i(h_j), $
	hence $\frac{t_j}{c_j \tilde{t}_j}=\frac{1}{\lc_i(h_j)}$.
	For any $j<m$, put $P_j = p_j - p_{j+1}$. We can then write:
	\begin{align*}
		P_j &= \frac{u}{v_j}\left(\frac{v_j}{u}p_j - \frac{v_j}{u}p_{j+1}\right) = \frac{u}{v_j}\left(\frac{t_jv_jh_j}{c_j\tilde{t}_j\lm_i(h_j)} - \frac{t_{j+1}v_jh_{j+1}}{c_{j+1}\tilde{t}_{j+1}\lm_i(h_{j+1})} \right)  \\
					  &= \frac{u}{v_j}\left(\frac{1}{\lc_i(h_j)}\frac{v_j}{\lm_i(h_j)}h_j - \frac{1}{\lc_i(h_{j+1})}\frac{v_j}{\lm_i(h_{j+1})}h_{j+1}\right)	\\
					  &= \frac{1}{\lc_i(h_j)\lc_i(h_{j+1})}\frac{u}{v_j}\left(\lc_i(h_{j+1})\frac{v_j}{\lm_i(h_j)}h_j - \lc_i(h_j)\frac{v_j}{\lm_i(h_{j+1})}h_{j+1}\right) \\
					  &= \frac{1}{\lc_i(h_j)\lc_i(h_{j+1})}\frac{u}{v_j}S(i,h_j,h_{j+1}).
	\end{align*}
	Plugging in the last expression back into equation \eqref{eqn:somme_telescopique} gives the desired equality with $d_j = \frac{e_j}{\lc_i(h_j)\lc_i(h_{j+1})}$ and $t^{\prime}_m = \frac{e_m}{c_m}t_m$.
	It satisfies \ref{enum:1}.
	The hypothesis forces $\val{}(e_m) > \val{}(c_m)$. Then we have $\valP(t^\prime_mh_m) = \val{}(e_m) + \valP(u) > \val{}(c_m) + \valP(u) = \valP(c_1u)$, which proves \ref{enum:2}.
	In addition, $\frac{u}{v_j} = k_j \in T_i$, which proves \ref{enum:3}.
	Finally, using that $\val (e_j) \geq c$ and $d_j = \frac{e_j}{\lc_i(h_j)\lc_i(h_{j+1})}$, one gets \ref{enum:4}.
\end{proof}

To prove the final Buchberger criterion in Proposition \ref{prop:buch_criterion}, we need the following lemma:
\begin{lemma}
	\label{lemma:inequality}
%	For $F$ a finite family in $\LaurentPolytopalRingForR$ and $(t_f)_{f \in  F}$ a family of terms, define $m((t_f)_{f \in F}) := \max(t_fd,\ d \in \textnormal{terms}(f),\ f \in  F)$.

%	Let $u \in T_i$ be such that $m((t_f)_{f \in  F}) < u$. 
%	Then for all $v \in T_i$, we have $m((vt_f)_{f \in  F}) < vu$.

If $f \in \LaurentPolytopalRingForR$ and $i \in \llbracket 1 , m \rrbracket$ are such that $\lt (f) <_r u$ for some $u \in T_i$,
then for any $v \in T_i,$ $ \lt (vf) <_r vu.$

\end{lemma}

\begin{proof}
%By hypothesis, we have $t_fd \le m((t_f)_{f \in F}) < u \in T_i$ for every term $t_fd$.
%Since $v \in T_i$, we have by (2) of Definition \ref{def:gmo} that $vt_fd < vu$ for every term $vt_f$, thus $m((vt_f)_{f \in F}) < vu$.  
Take $t$ a term of $f$. Then $t <_r u$ and $u,v \in T_i$ so 2. of Definition \ref{def:gmo} implies
$vt <_r vu.$ Taking the maximum on the $vt$'s, we conclude.
\end{proof}

\begin{proposition}[Buchberger criterion]
	\label{prop:buch_criterion}
Let $H = (h_1,\dots,h_m)$ be a family in $\LaurentPolytopalRingForR$ and $J$ the ideal generated by $H$.
For each $i \in I$ and $h_j \neq h_k \in H$, let $U(i,h_j,h_k)$ be a finite system of generators of the $T_i$-module $\lm_i(h_j)T_i(h_j) \cap \lm_i(h_k)T_i(h_k)$.
The following are equivalent:
\begin{enumerate}[(1)]
	\item $H$ is a Gröbner basis of $J$
	\item For all $i \in I$, $h_j \neq h_k$, $v \in U(i,h_j,h_k)$, $\rem(S(i,h_j,h_k,v),H) = 0$.
	\end{enumerate}
\end{proposition}
\begin{proof}
The fact that	(1) implies (2) is immediate from Prop \ref{prop:classical_easy}.

We now prove that (2) implies (1). 
By contradiction, assume that (2) is true and that $H$ is not a Gröbner basis of $J$. 
Then there exists $f \in  J$ such that $\lm(f) \notin \bigcup_{i \in I, 1\le j \le m}T_i(h_j)\lm_i(h_j)$.
Since, $f \in J = (h_1,\dots,h_m)$, we can write $f = \sum_{j=1}^{m}q_jh_j$ for some $q_j$ in $\LaurentPolytopalRingForR$. 
%By summing for each $j$ over the terms of $q_j$, we can rewrite $f$ as $\sum_{j=1}^{m}\sum_{\alpha \in \Delta(j)}t_{j,\alpha}h_j$ where for each $j$, $(t_{j,\alpha})_{\alpha \in \Delta(j)}$ is a finite family of terms. 
Write $\Delta(j)$ to be the set of terms of $q_j$. We can rewrite $f$ as $\sum_{j=1}^{m}\sum_{\alpha \in \Delta(j)}t_{j,\alpha}h_j$.
For such a writing of $f$, define $u = \max\{\lt(t_{j,\alpha}h_j), 1 \le j \le m, \alpha \in \Delta(j)\}$ and write the term $u$ as $u=c\tilde{u}$
for some $c \in K$ and some monomial $\tilde{u}.$
We have $\lt(f) <_r u$ because $\lm(f) \notin \cup_{i,j}T_i(h_j)\lm_i(h_j)$ and a cancellation has to appear.

Thus, $\valr{r}(u)$ is upper-bounded. Since $\val$ is discrete, there is a maximal $\valr{r}(u)$ among all
possible expressions of $f=\sum_{j=1}^{m}q_jh_j$. Among the expressions reaching this valuation, Lemma \ref{lemma:stric_descending_seq} ensure 
there is one such that $u$ is minimal. Let $i$ be such that $u \in T_i.$
Define $Z = \{(j,\alpha) \in \llbracket 1,m\rrbracket \times \Delta(j),\ \textrm{s.t. } \lt(t_{j,\alpha}h_j) = \gamma u,\ \gamma \in K, \ \val(\gamma)=0\}$
and $Z^\prime = \{(j, \alpha) \in \llbracket 1,m\rrbracket \times \Delta(j),\ \textrm{s.t. } \lt(t_{j,\alpha}h_j) <_r u\}$. 
We can then write:
\begin{equation}
 f = \sum_{(j,\alpha) \in Z}t_{j,\alpha}h_j + \sum_{(j,\alpha) \in Z^\prime}t_{j,\alpha}h_j \label{eq:ecriture_f_in_Buchberger}
\end{equation}
Let $g := \sum_{(j,\alpha) \in Z}t_{j,\alpha}h_j$.
We have $\lt(g) \le_r \max(\lt(f),\lt(\sum_{(j,\alpha) \in Z^\prime}t_{j,\alpha}h_j )) <_r u$ and $\lt(t_{j,\alpha}h_j) = c_{j,\alpha}\tilde{u}$ for all $(j,\alpha) \in Z$, where the $c_{j,\alpha}$ all have the same valuation.
So $g$ satisfies the conditions of Lemma \ref{lemma:sumSpair} and we can write
\begin{equation}
	g = \sum_{j=1}^{m-1}d_j\frac{\tilde{u}}{v_j}S(i,h_j,h_{j+1},v_j) + t^\prime_mh_m \label{eq:def_g_in_Buchberger}
\end{equation}
for some $d_j \in K$, $v_j \in U(i,h_j,h_{j+1})$, $\val \left( d_j \lc_i(h_j) \lc_i(h_{j+1}) \right) \geq \val (c)$ and $\tilde{u}/v_j \in T_i$ for $j<m$, and with $\lt(t^\prime_mh_m) < u$.
Now we use the hypothesis that all the S-pairs of elements of $H$ reduce to zero. 
For each $j<m$ we can write
\begin{equation}
S(i,h_j,h_{j+1},v_j) = \sum_{l=1}^{m} q^{(j)}_l  h_l,
\end{equation}
for some $q^{(j)}_l $'s in $\LaurentPolytopalRingForR$
satisfying 
\begin{align*}
\lt (q^{(j)}_l  h_l) & \leq_r \lt \left( S(i,h_j,h_{j+1},v_j) \right), \\
                    & <_r \lc_i(h_j)\lc_i(h_{j+1})v_j,
\end{align*}
where the last inequality comes from Lemma \ref{lemma:Spair}.
Since $v_j \in T_i$ and $\tilde{u}/v_j \in T_i$, we can apply
Lemma \ref{lemma:inequality}:

\begin{align*}
  \lt \left(\frac{\tilde{u}}{v_j}q^{(j)}_l  h_l\right) & <_r \lc_i(h_j)\lc_i(h_{j+1})v_j\frac{\tilde{u}}{v_j}  \\
 & = \lc_i(h_j)\lc_i(h_{j+1})\tilde{u}.
\end{align*}
Finally, using that $\val \left( d_j \lc_i(h_j) \lc_i(h_{j+1}) \right) \geq \val (c)$,
we deduce that for all $l \in \llbracket 1, m \rrbracket$ and
$j \in \llbracket1,m-1 \rrbracket$;
\[\lt \left(d_j \frac{\tilde{u}}{v_j}q^{(j)}_l  h_l\right)<_r u.\]
Inserting the expressions of $d_j \frac{\tilde{u}}{v_j}S(i,h_j,h_{j+1},v_j)$ as $\sum_{l=1}^{m} d_j \frac{\tilde{u}}{v_j}q^{(j)}_l  h_l$
in Equations \eqref{eq:def_g_in_Buchberger} and then \eqref{eq:ecriture_f_in_Buchberger},
we get an expression of $f$ in terms of the $h_j$'s with strictly smaller $u$, contradicting its minimality.
\end{proof}

\subsection{Buchberger's algorithm}
\label{subsec:Buchberger}

%\begin{proposition}
%	With the same notations as in Proposition \ref{prop:buch_criterion}, a Gröbner basis of $J$ can be computed the following way.
%	Construct an increasing sequence $(H_j)_{j \ge 0}$ by setting $H_0 := H$ and $H_{j+1} := H_j \cup R_j$ with
%	\[R_j := \bigcup_{i \in I} \{\rem(S(i,f,g,v)) \neq 0\ |\ f \neq g \in H_j,\ v \in U(i,f,g)\}\]
%	(The set $R_j$ can be computed using Algorithm \ref{alg:multi_div}).
%Then there exists an index $j$ for which $H_{j + 1} = H_{j}$ (i.e $R_{j} = \emptyset$), and $H_{j}$ is a Gröbner basis of $J$.
%\end{proposition}

\begin{proposition}
Algorithm \ref{alg:buchberger} on page \pageref{alg:buchberger} is correct and terminates, in the sense that it 
calls the multivariate $\textrm{division}$ a finite number of times.
\end{proposition}
\begin{proof}
    Correctness of the output is clear thanks to the Buchberger criterion of Proposition \ref{prop:buch_criterion}.
    For the termination, we prove that the addition of a new $s$ to $H$ (on Line 10) can only happen a finite amount of times.
    Indeed, let us assume that there is some input $J$ such that there is an infinite amount of non-zero $r$ happening on Line 8.
    Let $i \in I$ be an index such that there is an infinite amount of $\lt(s)$ in $T_i$ and let $H_j$ be an indexation
    of all the states of the set $H$ throughout Algo. \ref{alg:buchberger}.
    Using the second property of multivariate division in Prop. \ref{prop:multi_div}, 
    we can extract from the non-decreasing sequence $(<\lm_i(g)T_i(g),\ g \in H_j>_{K[T_i]})_{j \in  \N}$ of $K[T_i]$-ideals a strictly increasing one. 
	This is not possible since $K[T_i]$ is noetherian.
\end{proof}

%We provide explicit applications of Algorithm \ref{alg:buchberger} in the Annex.
We refer to the Appendix for a \sage demo and explicit examples over $\LaurentRing$.

\begin{algorithm}
	\SetKwInOut{Input}{input}\SetKwInOut{Output}{output}
	\caption{Buchberger algorithm in $\LaurentPolytopalRingForR$}
	\label{alg:buchberger}
	\Input{$J = (h_1,\dots,h_m)$ an ideal of $\LaurentPolytopalRingForR$}
	\Output{a Gröbner basis of $J$}
	\BlankLine
	$H \gets \{h_1,\dots,h_m\}$; $B \gets \{(h_i,h_j),\ 1 \le i < j\le m$\}

	\While{ $ B \neq \emptyset $}{
		$(f,g) \gets$ element of $B$; $B \gets B \setminus \{(f,g)\}$\;
		\For{ $ i \in I$ }{
			$U(i,f,g) \gets$ finite set of generators of $\lm_i(f)T_i(f)\cap  \lm_i(g)T_i(g)$ \tcp*{See Sec. \ref{sec:implem}} 
			\For{ $v \in U(i,f,g)$}{
				$\_,s \gets \textnormal{division}(S(i,f,g,v),H)$  \tcp*{Algo \ref{alg:multi_div}}
				\If{ $s \neq 0$}{
					$B \gets B \cup \{(h,s),\  h \in H\}$\;
					$H \gets H \cup \{s\}$
				}
			}
		}
	}
	\Return $H$
\end{algorithm}

\section{Implementation}
\label{sec:implem}
\subsection{Computation of the $T_i(f)$'s}
\label{subsec:computation_Ti}
To make our Buchberger algorithm fully explicit, we need a method to compute the finite set of generators $U(i, f, g)$ in line 5 of Algorithm \ref{alg:buchberger}. 
This problem was not addressed when g.m.o.'s for Laurent polynomials were introduced in \cite{PU:1999}.
Our idea is to compute the $T_i(f)$'s, for which we provide a general formula in Theorem \ref{th:compute_Tif}.

The following definitions are motivated by the fact that for particular g.m.o's (such as those of Examples \ref{ex:standard_gmo_1} to \ref{ex:second_gmo}),
the sets $A_i(f)$'s, $\Delta_{i,j}$'s and $U_i(f)$'s can be obtained by classical computations in polyhedral geometry.

\begin{definition}
    We define:
	\begin{enumerate}[(1)]
        \item $A_i(f):= \{t \in T,\ t\lm_i(f) \in T_i\}.$
        \item $\Delta_{i,j}(f):=\{t \in A_i(f) \cap A_j(f), \ t\lm_i(f) > t \lm_j(f)  \}.$
        \item $U_i(f):=A_i(f) \cap \bigcap_{j \in  I, \ \lm_i(f) \neq \lm_j(f)} \left( A_j(f)^c \cup \Delta_{i,j}(f) \right).$
    \end{enumerate}
\end{definition}

\begin{thm}
	\label{th:compute_Tif}
For any $i \in I,$ 
\begin{equation}
	\label{eqn:Ti_equal_Ui}
    T_i(f)=U_i(f) 
\end{equation}
\end{thm}
\begin{proof}
    We first prove that $T_i(f) \subset U_i(f).$
    Let $t  \in T_i(f).$ Then $\lm(tf) \in T_i$ so $t \lm_i(f)=\lm(tf) \in T_i $, hence $t \in A_i(f).$
    Let $j \in I$ be such that $\lm_i(f) \neq \lm_j(f)$.
    We can write that $t \lm_i(f)=\lm(tf) \geq t \lm_j(f)$ and since  $\lm_i(f) \neq \lm_j(f)$, this inequality becomes strict:  $t \lm_i(f)> t \lm_j(f)$. 
	Depending on whether $t \lm_j(f) \in T_j$ or not, $t$ is then in $A_j^c$ or $\Delta_{i,j}$.
    In conclusion, $t \in U_i(f)$ and $T_i(f) \subset U_i(f).$
    
    We prove the converse inclusion. 
    Let \[t \in U_i(f)=A_i \cap \bigcap_{j \in  I, \ \lm_i(f) \neq \lm_j(f)} \left( A_j^c \cup \Delta_{i,j} \right).\]
    From $t \in A_i,$ we get that $t \lm_i(f) \in T_i.$
    Let $j \in I$ be such that $\lm(tf) \in T_j$. Then $\lm(tf)=t\lm_j(f)\in T_j$ and $t \in  A_j(f).$
    If $j$ is such that $\lm_i(f)=\lm_j(f)$ then $t\lm_i(f)=t\lm_j(f)$, hence $\lm(tf)=t\lm_i(f)$ and  $\lm(tf) \in T_i$.
    Otherwise, our assumptions provide that $t \in U_i$ but not in $A_j(f),$ hence $t \in \Delta_{i,j}(f)$.
    Consequently, $ t\lm_i(f) > t \lm_j(f) =\lm(tf)$ which is in contradiction with the definition of $\lm(tf).$
    In conclusion, $\lm(tf)=t\lm_i(f) \in T_i$, hence $t \in T_i(f)$ and $T_i(f) = U_i(f).$
\end{proof}

% The $\lm_i$'s can be computed using the following Lemma:

% \begin{lemma}
%     If $t \in T$ is such that $t Supp(f) \subset T_i$ then $t \in T_i(f)$ and $\lm_i(f) = \frac{\lm(tf)}{t}.$
% \end{lemma}
% \begin{proof}
%     Let $t \in T$ be such that $t Supp(f) \subset T_i$.
%     Then $\lm(tf)=t x^\alpha$ for some $x^\alpha \in Supp(f).$ Thus $\lm(tf) \in T_i$ and $\lm(tf)=t \lm_i(f)$, from which we can conclude. 
% \end{proof}

If the $T_i$'s are defined as $C_i \cap \Z^n$ for $C_i$ a rational polyhedral cone, we can compute the $A_i$'s and $A_i^c$'s by solving systems of linear inequalities.
Moreover, if the g.m.o is defined by a function $\phi: T \to \Q_{\ge 0}$ as in Lemma \ref{lemma:construct_gto}, the $\Delta_{i,j}(f)$'s can be computed by computing first the $t$'s such that $\phi(t \lm_i(f)) > \phi(t \lm_j(f))$.
It obviously depends on $\phi$ but if the restriction of $\phi$ to each $T_i$ is linear, the desired set is again obtained by solving systems of linear inequalities.
For the case $\phi(t \lm_i(f)) = \phi(t  \lm_j(f))$, since our tie-break order $<_G$
is a group order, then $t\lm_i(f) > t\lm_j(f)$ if and only if $\lm_i(f) > \lm_j(f)$ which does not depend on $t$.
In total, one can compute generators of all the $T_i(f)$'s relying only on solving systems of linear inequalities from the following:
\begin{enumerate}
\item Computing the $\lm_i$'s,
\item Computing the $A_i(f)$'s and $A_i(f)^c$'s,
\item Computing the $\Delta_{i,j}(f)$'s (for $i,j$ such that $\lm_i(f) \neq  \lm_j(f)$),
\item Use formula \eqref{eqn:Ti_equal_Ui} to obtain the $T_i(f)$'s
\end{enumerate}

In the next part, we will make the computation of the $T_i(f)$'s completely explicit for g.m.o's whose underlying conic decomposition is the decomposition of Example \ref{ex:standard_conic_decomposition}.

\subsection{A particular conic decomposition}
\label{subsec:particular_conic_decompo}
Taking a closer look at Algorithm \ref{alg:buchberger}, we can identify two areas for minimizing the number of S-pairs to be reduced:
\begin{enumerate}[(1)]
\item Within the for loop at line 4, by minimizing the value of $|I|$ (the number of cones).
\item Within the for loop at line 6, by specifying that the set of generators $U(i, f, g)$ should contain only one element.
\end{enumerate}
The first point is addressed comprehensively in the following proposition, while Theorem \ref{th:monogenous} offers a partial solution to the second point.

\begin{proposition}
	\label{th:minimal_number_cones}
	A conic decomposition of $T\cong \Z^n$ contains at least $n+1$ cones.
\end{proposition}
\begin{proof}
	Assume, for the sake of contradiction, the existence of a conic decomposition containing strictly less than $n+1$ cones.
	Without loss of generality, we may assume that it contains exactly $n$ cones $T_1,\dots,T_n$.

	By (1) of Definition \ref{def:conic_decomposition}, for each $1 \le i \le n$, $T_i$ is contained in a cone $C_i := a_{i,1}\R_{\ge 0} \oplus \dots \oplus a_{i,n}\R_{\ge 0}$ for some $a_{i,1},\dots,a_{i,n} \in \Z^n$ such that $(a_{i,1},\dots,a_{i,n})$ is a basis of $\R^n$,
	and by (2) of Definition \ref{def:conic_decomposition}, $\cup T_i = \Z^n$. 
	Thus $\cup T_i = \Z^n$ is contained in the closed subset $C := \cup C_i$ of $\R^n$. 
	We show that $\Z^n \setminus C \neq \emptyset$, resulting in a contradiction.

	We first show that $\R^n \setminus C \neq \emptyset$.
	For each $i$ there exists $v_i \in \R^n \setminus \{0\}$ such that $C_i \setminus \{0\}$ is contained in the linear open half-space $L(v_i) := \{ u \in \R^n|\ u \cdot v_i < 0\}$, and there exists an open neighborhood $E_i$ of $v_i$ such that the inclusion $C_i \subset L(w_i)$ is still true for $w_i \in  E_i$.
	Thus we can choose the $v_i$'s such that $(v_1,\dots,v_n)$ is a basis of $\R^n$.
	By Gram-Schmidt process, there exists an orthonormal basis $b_1,\dots,b_n$ such that the matrice of $(v_1,\dots,v_n)$ in that basis is upper triangular with strictly positive coefficients on the diagonal.
	By construction, the vector $b_n$ satisfies $ b_n \cdot v_i \ge 0$ for each $i$. 
	Thus $b_n \in \R^n \setminus \cup_i L(v_i) \subset \R^n \setminus \cup_i C_i = \R^n \setminus C$, and so $\R^n \setminus C \neq \emptyset$.

	Since $\R^n \setminus C$ is a non-empty open set, it contains a basis $(c_1,\dots,c_n)$ of $\R^n$, and thus contains the cone $c_1\R_{\ge0}\oplus \dots \oplus c_n\R_{\ge 0}$ which
	itself contains an infinity of elements of $\Z^n$.
\end{proof}

\begin{lemma}
	\label{lemma:conic_group}
	Let $\le$ be a g.m.o such that the underlying conic decomposition $(T_i)_{i \in I}$ satisfies:
	\[
		\forall i,j \in I,\ _\textnormal{gr} \langle T_i\cap T_j \rangle \cap T_i = T_i\cap T_j,
	\]
	where $_\textnormal{gr} \langle T_i\cap T_j \rangle$ is the group generated by the monoid $T_i \cap T_j$.
	For all $i,j \in I, f \in \LaurentRing$, we have:
	\[
		(s \in T_i \cap T_j,\ t \in T_i(f),\ st \in T_i(f)\cap T_j(f)) \implies t \in T_j(f).
	\]
\end{lemma}

\begin{proof}
   Writing $l := \lm_i(f) = \lm_j(f)$, we have $\lm(tf) = tl \in T_i$ and $\lm(stf) = stl \in T_i \cap T_j$ because $t \in T_i(f)$ and $st \in  T_i(f)\cap T_j(f)$. 
   We have to show that $tl \in  T_j$. 
   Now $tl = s^{-1}stl \in {_\textnormal{gr}}\langle T_i \cap T_j \rangle \cap T_i = T_i \cap T_j$ by the hypothesis.
\end{proof}

The conic decompositions of Examples \ref{ex:standard_conic_decomposition} and \ref{ex:second_conic_decomposition} satisfy the hypothesis of Lemma \ref{lemma:conic_group}

% \begin{proposition}
% 	Let $\le$ be a g.m.o such that the underlying conic decomposition is the decomposition $(T_0,T_1,\dots,T_n)$ of Example \ref{ex:standard_conic_decomposition}. 
% 	For all $f \in \LaurentRing$  and $0 \le i \le n$, $T_i(f)$ is a monogenous $T_i$-module.
% \end{proposition}
% \begin{proof}
% Reasoning by contradiction, let's assume that there exists an index $i$ and $f \in \LaurentRing$ such that the set of minimal generators of $T_i(f)$ contains at least $2$ elements $a \neq b \in \Z^n$. 

% We only prove the case $i=0$, the proof for $i \neq 0$ being essentially the same. We have $a + T_0 = \cap_j\{x \in \Z^n,\  x_j \ge a_j\}$ and $b + T_0 = \cap_j\{x \in \Z^n,\  x_j \ge b_j\}$.
% % Since $a \neq b$, there is at least one index $j$ such that $a_j \neq, b_j$, and we may suppose that $a_j > b_j$.
% The vector $p := \min(a,b) \in \Z^n$ is different from $a$ and $b$, and it satisfies $a + T_0 \cup b + T_0 \subset p + T_0$. 
% Thus $p \notin T_0(f)$, otherwise $\{a,b\}$ wouldn't be contained in a minimal set of generators of $T_0(f)$.
% Since the $T_i(f)'s$ cover $\Z^n$, there exists $k \neq 0$ such that $ p \in T_k(f)$. 
% We have $T_0 \cap T_k = \{x_k = 0\} \cap (\cap_{j \neq k}\{ x_j \ge 0\})$,
% thus the set $p + T_0 \cap T_k $ contains either $a$ or $b$, depending on whether $\min(a_k,b_k) = a_k$ or $b_k$.
% We may suppose that it contains $a$.
% Finally, we have $p = a + (p-a)$ with $p \in T_k(f)$, $(p-a) \in T_k \cap T_0$ and $a \in T_k(f) \cap T_0(f)$.
% By Lemma \ref{lemma:conic_group}, $p \in T_0(f)$, contradicting $p \notin T_0(f)$.
% \end{proof}

\begin{thm}
	\label{th:monogenous}
Let $\le$ be a g.m.o such that the underlying conic decomposition is $(T_0, T_1, \dots, T_n)$ as defined in Example \ref{ex:standard_conic_decomposition}. For all $f \in \LaurentRing$ and $0 \leq i \leq n$, $T_i(f)$ is a monogenous $T_i$-module.
\end{thm}

\begin{proof}
Assume, for the sake of contradiction, that there exists an index $i$ and $f \in \LaurentRing$ such that the set of minimal generators of $T_i(f)$ contains at least two elements $a \neq b \in \Z^n$.
We will focus on the case where $i=0$; the proof for $i \neq 0$ is essentially the same. 
We have $a + T_0 = \bigcap_j\{x \in \Z^n \mid x_j \geq a_j\}$ and $b + T_0 = \bigcap_j\{x \in \Z^n \mid x_j \geq b_j\}$.
Let $p \in \Z^n$ be the vector such that $p_j = \min(a_j, b_j)$ for $1 \le j \le n$, which is necessarily different from $a$ and $b$. By definition of $p$, we have $a + T_0 \cup b + T_0 \subset p + T_0$.
Therefore, $p \notin T_0(f)$; otherwise, $\{a, b\}$ wouldn't be contained in a minimal set of generators of $T_0(f)$.
Since the $T_i(f)$'s cover $\Z^n$, there exists $k \neq 0$ such that $p \in T_k(f)$.
We have $T_0 \cap T_k =\{x_k = 0\} \cap \left(\bigcap_{j \neq k}\{ x_j \geq 0\}\right)$. Consequently, the set $p + T_0 \cap T_k $ contains either $a$ or $b$, depending on whether $\min(a_k, b_k) = a_k$ or $b_k$. Without loss of generality, assume it contains $a$.
From $a \in p + T_0 \cap T_k,$ we deduce that $a-p \in T_0 \cap T_k.$
Also, from $a \in p + T_0 \cap T_k \subset p+T_k$ and  $p \in T_k(f)$, we get that $a \in T_k(f)$ and thus $a \in T_k(f) \cap T_0(f)$.
Consequently we can write $(a-p)+p=a$ with $(a-p) \in T_k \cap T_0$, $p \in T_k(f)$ and $a \in T_k(f) \cap T_0(f)$. According to Lemma \ref{lemma:conic_group}, this implies $p \in T_0(f)$, contradicting $p \notin T_0(f)$.
\end{proof}

Altogether, Proposition \ref{th:minimal_number_cones} and Theorem \ref{th:monogenous} assert that, for a g.m.o defined over the decomposition of Example \ref{ex:standard_conic_decomposition}, the number of S-pairs to be reduced in Algorithm \ref{alg:buchberger} is minimized to the greatest extent possible.
For g.m.o's of this nature,  we can determine the generator $g_i$ of $T_i(f)$ using a straightforward descent algorithm.
Recall that $T_i$ is the monoid generated by $H_i = \{e_1,\dots,\hat{e_i},\dots,e_n\} \cup \{-(e_1+\dots+e_n)\}$, where $(e_1,\dots,e_n)$ denote the canonical basis of $\R^n$.
We first find a $t \in T_i(f)$ so that $T_i(f)=g_i + T_i \subset t + T_i$. Then, for each generator $h \in H_i$ of $T_i$, we compute $t = t - h$ until $t + T_i \not\subset T_i(f)$, for then we set $t = t +h$. At the end, $t = g_i$. The procedure is detailed in Algorithm \ref{alg:generator}.

\begin{remark}
	For a g.m.o as in Theorem \ref{th:monogenous}, it can easily be shown that if $g_j, g_k$ are the generators of $T_j(f), T_k(f)$ respectively, then $\lVert g_j - g_k \rVert_\infty \le 1$ (with $\Vert \cdot \Vert_\infty$ the maximum norm on $\R^n$). 
	To speed up the computation of the $T_i(f)$'s, we can first determine $g_0$ using Algorithm \ref{alg:generator}.
	Then the generators of the $T_i(f)$'s for $ i \in \llbracket1,n\rrbracket$ are contained within the (small) set $\{y \in \Z^n,\ \lVert y - g_0 \rVert_\infty \le 1\}$.
\end{remark}

\begin{remark}

   If in addition the g.m.o is defined by a function as outlined in Lemma \ref{lemma:construct_gto}, and this function restricts to an integral linear function on each cone (as illustrated in Examples \ref{ex:standard_gmo_1} and \ref{ex:standard_gmo_2}), then the generator of $T_i(f)$ can be expressed directly as the intersection of $n$ affine hyperplanes defined by integral equalities.
\end{remark}

We emphasize that the time spent computing the generators of the $T_i(f)$'s is negligible compared to the time spent reducing S-pairs in Algorithm \ref{alg:buchberger}.

\begin{algorithm}
	\SetKwInOut{Input}{input}\SetKwInOut{Output}{output}
	\caption{Generator of $T_i(f)$}
	\label{alg:generator}
	\Input{ $0 \le i \le n$, $f \in  \LaurentPolytopalRingForR$, $H_i = \{h_1,\dots,h_n\}$ generators of $T_i$, and the g.m.o. of Example \ref{ex:standard_conic_decomposition}}
	\Output{ $g \in \LaurentMonoid$ such that $T_i(f) = g + T_i$ }
	$t \gets$ a monomial such that all monomials of $t \times in_r(f)$ are in $T_i$ \tcp*{Use Remark \ref{rem:how_to_compute_the_Tis}}
	\For{ $h \in H_i$}{
		\While{ $t \in  T_i(f)$}{
			$t \gets t -h$\;
		}
		$t \gets t + h$\;
	}
	\Return{t}
	\BlankLine
\end{algorithm}

\newpage
\section{Gröbner bases for ideals in $ \LaurentPolytopalRing $}
\label{sec:gb_valP}

In the previous sections, we developed a Gröbner theory for an ideal \( J \subseteq \LaurentPolytopalRingForR \). 
Here, we extend this framework to a general \( P \) that may not contain just a single point.
In \( \LaurentPolytopalRingForR \), the valuation \( \val_{r} \) satisfies the property:
\[
\val_{r}(af) = \val_{r}(a) + \val_{r}(f),
\]
which implies
\[
\textnormal{in}_{r}(af) = a \cdot \textnormal{in}_{r}(f).
\]
% This additivity was crucial in our previous constructions (see ...).
However, in \( \LaurentPolytopalRing \), the valuation \( \val_{P} \) only satisfies the inequality:
\[
\val_{P}(af) \geq \val_{P}(a) + \val_{P}(f).
\]
Simply replacing \( \val_{r} \) with \( \val_{P} \) in Definition~\ref{def:order} may lead to:
\[
\textnormal{in}_{P}(af) \neq a \cdot \textnormal{in}_{P}(f),
\]
as illustrated in the following example:

\begin{example}
\label{ex:changement_in}
Let $ r_1 = (1,1), r_2 = (0,1) $, $P = \textnormal{Conv}(\{r_1,r2\})$, $K = \mathbb{Q}_{2}$, $a = (xy)^{-1}$, and $f = 2x + y \in K\{x^{\pm 1}, y^{\pm 1}; P\}$.  
    We have $\textnormal{in}_{P}(f) = y$, but $\textnormal{in}_{P}(af) = a \cdot 2x \neq a \cdot \textnormal{in}_{P}(f)$.
\end{example}

In Example \ref{ex:changement_in}, the key issue is that \( \val_P(a) \) is attained at \( r_2 \) and not at $ r_1 $,
whereas \( \val_P(f) \) is attained at \( r_1 \) and not at $ r_2 $.
On the other hand, if both \( \val_P(a) \) and \( \val_P(f) \) are attained at
the same vertex \( r_i \), it follows that \( \val_P(af) \) is also attained at \( r_i \).

This observation motivates the introduction of the sets \( V_i \) in Definition~\ref{def:Uibar},
which consist of the monomials \( m \) such that \( \val_P(m) = \val_{r_i}(m) \).
However, \( \val_P(m) \) may be attained at more than one vertex \( r_i \).
To resolve such ties, we impose an arbitrary ordering on the vertices of \( P \)
and introduce an intermediate step (see Definition~\ref{def:new_order}) to order terms
accordingly.

Finally, since the last tie-breaking step (using a generalized order \( \leq_\omega \)) occurs \emph{after}
considering the valuation and the order on the indices, we need to ensure that each cone in 
the conic decomposition underlying \( \leq_\omega \) is fully \emph{contained} within some \( V_i \).
We say that such a conic decomposition respects the \( V_i \)'s.
Figures 2-5 illustrates this notion.
\smallbreak
With these adapted settings, we prove Proposition~\ref{prop:lmij_unicity},
which generalizes Proposition~\ref{lemma:lmi_independent} and Proposition \ref{prop:Uifg}.
Then, the other results of Gröbner theory extend (almost) unchanged.

Naturally, applying the results of this section to the special case \( P = \{r\} \)
(a single vertex) recovers the previous theory exactly.

% What happens in Example \ref{ex:changement_in} is that $ \val_P(a) $ is reached at $ r_2 $ and not at $ r_1 $, while
% $ \val_P(f) $ is reached at $ r_1 $ but not at $ r_2 $.
% One the other hand, one can easily check that if $ \val_P(a) $ and $ \val_P(f) $
% are reached at the same verex $ r_i $, then it is true again for $ \val_P(af) $.
% This motivates the introduction of the sets
% $ V_i $ in Definition \ref{def:Uibar} which are exactly the monomials $ m $ for which $ \val_P(m) = \val_{r_i}(m)$.
% Now it may happen that $ \val_P(f) $ is reached for mor than one vertex $ r_i $.
% To break-ties in this case, we fix an (arbitrary) order on the vertices of $ P $ and 
% add an intermediate step (see Definition \ref{def:new_order}) when orderings terms taking into acocunt
% this indexing.
% Finally, since the last breakin ties step (using a generalized ordre $ \le_\omega $) happens after
% breaking by the valuation, it is necessarily that each cone int the underlying conic decomposition of
% $ \le_\omega $ is included in one $ V_I $. We say that such a conic decomposition respects
% the $ v_i $.
% With the adapted settings, we proove Propsition \ref{th:lmij_unicity}, which 
% generalizes in this case Prop and Prop, and carries on the other propposition of grobner theory almost the same way.
%
% Of course, applying thes reuslts of this section to $ P = {r} $ for a single vertex, 
% we fall back exactly to the previous case.

\subsection{New definitions}

Let us assign an (arbitrary) indexing to the vertices of $P$: $\textnormal{vert}(P) = \{r_1, \dots, r_t\}$, 
and let $I_P$ represent the set of indexing indices: $I_P := \{1, \dots, t\}$. 
This indexing will be used to resolve ties in Definition \ref{def:new_order} when $\val_P(f) = \val_{r_i}(f)$ 
for more than one $r_i$.

\begin{definition}
\label{def:I_P(X)}
For $ f \in \MonoidOfTerms $, we define
\[ I_P(f) := \{ i \in I_P ,\ \val_P(f) = \val_{r_i}(f)\} \subseteq I_P.  \]
\end{definition}
Thus, $I_P(f) \subseteq I_P$ is the subset of indices at which $\val_P(f)$ is attained.

\begin{example}
	\label{ex:I_Pf}
	Set $ r_1 = (1,2), r_2 = (2,1), r_3 = (0,0) $.
Take $ P = \textnormal{Conv}(\{r_1,r_2,r_3\})$, $K =  \Q_2 $ and 
	$ f = \frac{1}{2}x + y^{2} + xy \in K \{x^{\pm 1}, y^{\pm 1}; P\}$.
	We have $ \val_P(f) = \min_{i \in I_P}\val_{r_i}(f) = \val_{r_1}(f) = \val_{r_2}(f) = -3$, so $ I_P(f) = \{1,2\}  $.
	This common minimum is reached at $ xy $ for $ \val_{r_1} $ and at $ \frac{1}{2}x $ and $xy$ for $ \val_{r_2} $.
\end{example}

\begin{example}
\label{ex:I_Pmonomial}
	Set $ r_1 = (0,0,3), r_2 = (1,0,0), r_3 = (-1,1,2), r_4 = (0,1,1) $.
Take $ P = \textnormal{Conv}(\{r_1,r_2,r_3,r_4\})$, $K =  \Q_2 $ and 
$ t = 2xyz^{-1} \in K \{x^{\pm 1}, y^{\pm 1}, z^{\pm 1}; P\}$.
	We have $ \val_P(t) = \min_{i \in I_P}\val_{r_i}(t) = \val_{r_2}(t) = \val_{r_4}(t) = 0$, so $ I_P(t) = \{2,4\}  $.
	Since $ t $ is a single term, this common minimum is necessarily reached at $ t $.
\end{example}

\begin{definition}
\label{def:Uibar}
For each $i \in I_P$, we define:
\begin{itemize}
    \item The monoid $V_i \subseteq T$ by: 
    \begin{align*}
		V_i &=\{ \X^\alpha \in T,\ i \in I_P(\X^{\alpha})\} \\
		    &=\{ \X^\alpha \in T,\ \val_{P}(\X^\alpha) = \min_{j \in I_P}(\val_{r_j}(\X^{\alpha})) = \val_{r_i}(\X^{\alpha})\} \\
			&=\{ \X^\alpha \in T,\ r_i \cdot \alpha \leq r_j \cdot \alpha ,\ \forall j \in I_P \}.
    \end{align*}
		% $V_i := \{ \X^{\alpha} \in T ,\ i \in I_p(\X^{\alpha})\}.$ 

	\item The subset $ V_{i,<} \subseteq  V_i $ by:
	\begin{align*}
		V_{i,<} &= \{ \X^\alpha \in T ,\  \min(I_P(\X^\alpha ) = i \} \\
				&= \left\{ \X^{\alpha} \in T,\ 
				\begin{cases}
					r_i \cdot \alpha < r_j \cdot \alpha \textnormal{ if }  j < i  \\ 
					r_i \cdot \alpha \le r_j \cdot \alpha \textnormal{ if } i \le j  
				\end{cases}, \forall j \in I_P \right\}.
	\end{align*}
	$V_{i,<}$ is obtained from $V_i$ by replacing each inequality involving $j < i$ with a strict inequality, hence the sign $ < $ in the notation.
\end{itemize}
\end{definition}

By construction, we have $ \cup_{i \in I_P} V_i = T $, 
the sets $ V_{i,<} $ are pairwise disjoint, and $ \coprod_{i \in I_P} V_{i,<} = T $. 
To aid in visualizing $ V_i $ and $ V_{i,<} $, 
we introduce their equivalents in $ \R^n $, $ C_i $ and $ C_{i,<} $:
\[
	\begin{split}
	C_{i} := \{ \alpha \in \R^n ,\ i \in I_P(\X^{\alpha})\} &\quad \quad C_{i,<} := \{ \alpha \in \R^n ,\ i= \min(I_P(\X^{\alpha})) \} 
	\end{split}
\]

That way, we have $ V_i = \{ \X^{\alpha},\ \alpha \in C_{i}\cap \Z^n\}$,
$ V_{i,<} = \{ \X^{\alpha} ,\ \alpha \in C_{i,<}\cap\Z^n\}$, and so
the monoid $ V_i $ (resp. the set $ V_{i,<} $) can be identified with the monoid of 
integer-coordinate points in $C_i$ (resp. the set of integer-coordinate points in $ C_{i,<}$). 

\bigbreak

Now consider a conic decomposition such that each cone is contained within a single \( V_i \).
For each \( i \in I_P \), let \( D_i \) be the indexing set containing the indices of all cones contained in \( V_i \). 
Hence for each \( V_i \), the cones contained in $ V_i $, noted \( (T_{i,j})_{j \in D_i} \), is a conic decomposition of the monoid \( V_i \).
and the collection 
\[
\{T_{i,j} \mid i \in I_P, j \in D_i\}
\] 
forms a conic decomposition of \( T \).
An alternative method to construct such a conic decomposition is to take the union of the traces of a conic decomposition of \( T \) on each \( V_i \). 
We provide four visual examples in Figures~\ref{fig:decompo_point}--\ref{ex:decompo_full_bad} on page \pageref{fig:decompo_point}-\pageref{ex:decompo_full_bad}. 
In each case:
\begin{itemize}
    \item The left figure represents the polytope $P$,  
    \item The middle figure illustrates the cones $C_i$ (and thus $V_i$), and  
    \item The right figure shows the decomposition of each $V_i$ into the various $T_{i,j}$ (dashed lines).
	\item The sets $ C_{i,<} $ (and thus $ V_{i,<} $) are not directly shown on Figures ~\ref{fig:decompo_point}--\ref{ex:decompo_full_bad}, but
	can be obtained by removing from $ C_i $ all the facets $ C_i \cap C_j $ for $ j < i $.
\end{itemize}

\begin{remark}
    \label{re:different_breakings}
    We make a few remarks on Figures~\ref{fig:decompo_point}--\ref{ex:decompo_full_bad}.
    \begin{itemize}
        \item In Figure \ref{fig:decompo_point}, $C_1 \cap \Z^n = \Z^n$, thus $V_1 \cap T = T$, and $\{T_{1,1}, T_{1,2}, T_{1,3}\}$ is simply a conic decomposition of $T$. This corresponds to the case studied in the first part of the paper.  
        \item In Figure \ref{fig:decompo_segment}, the monoids $V_1$ and $V_2$ are not conic decompositions of themselves (each contains non-trivial invertible elements), so we must break each of them into smaller pieces (two in this case).
        \item In Figure \ref{ex:decompo_full}, all the $V_i$ are already conic decompositions of themselves, so we can simply set $T_{i,1} = V_i$ for each $i$.  
        \item In Figure \ref{ex:decompo_full_bad}, we reuse the polytope from Figure 4 but further decompose the monoids $V_1$ and $V_3$ into two pieces each. While this is possible, it is not necessary.  
        In general, we aim to minimize the number of cones in our conic decomposition. 
        Thus, we decompose each $V_i$ into smaller pieces only if $V_i$ is not already a conic decomposition of itself, 
        as in Figures \ref{fig:decompo_point} and \ref{fig:decompo_segment}, and avoid unnecessary decompositions like those shown in Figure \ref{ex:decompo_full_bad}.
    \end{itemize}
\end{remark}

\begin{figure}[h]
	\centering
	\subfloat[ee]{
		\centering
		\begin{tikzpicture}[scale=0.4]
    \definecolor{vertex2}{RGB}{255,0,0}
	% Grille de points entiers
	\foreach \x in {-3,-2,-1,0,1,2,3}
		\foreach \y in {-3,-2,-1,0,1,2,3}
		{
			\fill[black!30] (\x,\y) circle (0.05);
		}
        % Axes
        \draw[->, opacity=0.4] (-3,0) -- (3,0) node[right] {$x$};
        \draw[->, opacity=0.4] (0,-3) -- (0,3) node[above] {$y$};
        \fill[vertex2, opacity=0.6] (1,1) circle (0.2) node[above right, opacity=0.6] {$r_1$};
	\end{tikzpicture}
}
	\hfill
	\subfloat[eee]{
		\centering
		\begin{tikzpicture}[scale=0.4]
    \definecolor{vertex2}{RGB}{255,0,0}
	% Grille de points entiers
	\foreach \x in {-3,-2,-1,0,1,2,3}
		\foreach \y in {-3,-2,-1,0,1,2,3}
		{
			\fill[black!30] (\x,\y) circle (0.05);
		}
        % Axes
        \draw[->, opacity=0.4] (-3,0) -- (3,0) node[right] {$x$};
        \draw[->, opacity=0.4] (0,-3) -- (0,3) node[above] {$y$};

		\fill[vertex2,opacity=0.6] (0,0) circle (3);
		% \draw[dashed, thick] (0,0) -- (3,0);
		% \draw[dashed, thick] (0,0) -- (0,3);
		% \drawVector{-1}{-1}{2.2}{dashed, thick};

		% \node at (1.5,1.5) {$ T_{1,1}$};
		% \node at (1,-1.5) {$ T_{1,2}$};
		% \node at (-1.5,1) {$ T_{1,3}$};
		\node[vertex2, opacity=0.6] at (2.5,2.5) {$ C_1 $};
	\end{tikzpicture}
}
\hfill
\subfloat{
		\centering
		\begin{tikzpicture}[scale=0.4]
    \definecolor{vertex2}{RGB}{255,0,0}
	% Grille de points entiers
	\foreach \x in {-3,-2,-1,0,1,2,3}
		\foreach \y in {-3,-2,-1,0,1,2,3}
		{
			\fill[black!30] (\x,\y) circle (0.05);
		}
        % Axes
        \draw[->, opacity=0.4] (-3,0) -- (3,0) node[right] {$x$};
        \draw[->, opacity=0.4] (0,-3) -- (0,3) node[above] {$y$};

		\fill[vertex2,opacity=0.6] (0,0) circle (3);
		\draw[dashed, thick] (0,0) -- (3,0);
		\draw[dashed, thick] (0,0) -- (0,3);
		\drawVector{-1}{-1}{3}{dashed, thick};

		\node at (1.5,1.5) {$ T_{1,1}$};
		\node at (1,-1.5) {$ T_{1,2}$};
		\node at (-1.5,1) {$ T_{1,3}$};
		\node[vertex2, opacity=0.6] at (2.5,2.5) {$ C_1 $};
	\end{tikzpicture}
}
\caption{}	\label{fig:decompo_point}
\end{figure}

\begin{figure}[h]
	\centering
	\subfloat[ee]{
		\centering
		\begin{tikzpicture}[scale=0.4]
	\foreach \x in {-3,-2,-1,0,1,2,3}
		\foreach \y in {-3,-2,-1,0,1,2,3}
		{
			\fill[black!30] (\x,\y) circle (0.05);
		}
        % Axes
        \draw[->, opacity=0.4] (-3,0) -- (3,0) node[right] {$x$};
        \draw[->, opacity=0.4] (0,-3) -- (0,3) node[above] {$y$};

		\coordinate (v1) at (1,1);
		\coordinate (v2) at (-2,-1);

		\draw[thick] (v1) -- (v2);
        \fill[vertex2, opacity=0.6] (v1) circle (0.2) node[above right] {$r_1$};
        \fill[vertex3, opacity=0.6] (v2) circle (0.2) node[below left] {$r_2$};
	\end{tikzpicture}
}
	\hfill
	\subfloat[eee]{
		\centering
		\begin{tikzpicture}[scale=0.4]
    \definecolor{vertex2}{RGB}{255,0,0}
	% Grille de points entiers
	\foreach \x in {-3,-2,-1,0,1,2,3}
		\foreach \y in {-3,-2,-1,0,1,2,3}
		{
			\fill[black!30] (\x,\y) circle (0.05);
		}
        % Axes
        \draw[->, opacity=0.4] (-3,0) -- (3,0) node[right] {$x$};
        \draw[->, opacity=0.4] (0,-3) -- (0,3) node[above] {$y$};

		% \fill[vertex2,opacity=0.6] (0,0) circle (3);
		% \draw[dashed, thick] (0,0) -- (3,0);
		% \draw[dashed, thick] (0,0) -- (0,3);
		% \drawVector{-1}{-1}{2.2}{dashed, thick};
		% \draw[thick] (2,-3) -- (-2,3);
		\drawVector{2}{-3}{3}{thick};
		\drawVector{-2}{3}{3}{thick};
		\fill[vertex2, opacity=0.6] (0,0) -- (124:3) arc (124:303:3);
		\fill[vertex3, opacity=0.6] (0,0) -- (-57:3) arc (-57:124:3);

		% \node at (1.5,1.5) {$ T_{1,1}$};
		% \node at (1,-1.5) {$ T_{1,2}$};
		% \node at (-1.5,1) {$ T_{1,3}$};
		\node[vertex3, opacity=0.6] at (2.5,2.5) {$ C_2 $};
		\node[vertex2, opacity=0.6] at (-2.5,-2.5) {$ C_1 $};
	\end{tikzpicture}
}
\hfill
\subfloat[eee]{
		\centering
		\begin{tikzpicture}[scale=0.4]
	\foreach \x in {-3,-2,-1,0,1,2,3}
		\foreach \y in {-3,-2,-1,0,1,2,3}
		{
			\fill[black!30] (\x,\y) circle (0.05);
		}
        % Axes
        \draw[->, opacity=0.4] (-3,0) -- (3,0) node[right] {$x$};
        \draw[->, opacity=0.4] (0,-3) -- (0,3) node[above] {$y$};

		% \fill[vertex2,opacity=0.6] (0,0) circle (3);
		% \draw[dashed, thick] (0,0) -- (3,0);
		% \draw[dashed, thick] (0,0) -- (0,3);
		% \drawVector{-1}{-1}{2.2}{dashed, thick};
		% \draw[thick] (2,-3) -- (-2,3);
		\drawVector{2}{-3}{3}{dashed,thick};
		\drawVector{-2}{3}{3}{dashed,thick};
		\fill[vertex2, opacity=0.6] (0,0) -- (124:3) arc (124:303:3);
		\fill[vertex3, opacity=0.6] (0,0) -- (-57:3) arc (-57:124:3);

		% \node at (1.5,1.5) {$ T_{1,1}$};
		% \node at (1,-1.5) {$ T_{1,2}$};
		% \node at (-1.5,1) {$ T_{1,3}$};
		\node[vertex3, opacity=0.6] at (2.5,2.5) {$ C_2 $};
		\node[vertex2, opacity=0.6] at (-2.5,-2.5) {$ C_1 $};
		\drawVector{3}{2}{3}{dashed, thick};
		\drawVector{-3}{-2}{3}{dashed, thick};
		
		\node at (-1.8,0.5) {$ T_{1,1} $};
		\node at (-0.5,-2) {$ T_{1,2} $};
		\node at (0.5,2) {$ T_{2,1} $};
		\node at (2,-0.8) {$ T_{2,2} $};

	\end{tikzpicture}
}
\caption{}\label{fig:decompo_segment}
\end{figure}

\begin{figure}
	\centering
	\subfloat[ee]{
		\centering
		\begin{tikzpicture}[scale=0.4]
	\foreach \x in {-3,-2,-1,0,1,2,3}
		\foreach \y in {-3,-2,-1,0,1,2,3}
		{
			\fill[black!30] (\x,\y) circle (0.05);
		}
        % Axes
        \draw[->, opacity=0.4] (-3,0) -- (3,0) node[right] {$x$};
        \draw[->, opacity=0.4] (0,-3) -- (0,3) node[above] {$y$};

        \coordinate (A) at (-2,2);
        \coordinate (B) at (1,2);
        \coordinate (C) at (2,-2);
        \coordinate (D) at (-1,-1);
        
        % Arêtes du polytope
        \draw[thick] (A) -- (B) -- (C) -- (D) -- cycle;
        
        % Points colorés
        \fill[vertex1, opacity=0.6] (A) circle (0.2) node[above left] {$r_1$};
        \fill[vertex2, opacity=0.6] (B) circle (0.2) node[above right] {$r_2$};
        \fill[vertex3, opacity=0.6] (C) circle (0.2) node[below right] {$r_3$};
        \fill[vertex4, opacity=0.6] (D) circle (0.2) node[below left] {$r_4$};
	\end{tikzpicture}
}
	\hfill
	\subfloat[eee]{
		\centering
		\begin{tikzpicture}[scale=0.4]
    \definecolor{vertex2}{RGB}{255,0,0}
	% Grille de points entiers
	\foreach \x in {-3,-2,-1,0,1,2,3}
		\foreach \y in {-3,-2,-1,0,1,2,3}
		{
			\fill[black!30] (\x,\y) circle (0.05);
		}
        % Axes
        \draw[->, opacity=0.4] (-3,0) -- (3,0) node[right] {$x$};
        \draw[->, opacity=0.4] (0,-3) -- (0,3) node[above] {$y$};

		\fill[vertex1, opacity=0.6] (0,0) -- (-90:3) arc (-90:18:3) -- cycle;
		\node[vertex1, opacity=0.6] at (2.5, -2.5) {$ C_1 $};
		\fill[vertex4, opacity=0.6] (0,0) -- (18:3) arc (18:72:3) -- cycle;
		\node[vertex4, opacity=0.6] at (2.5, 2.5) {$ C_4 $};
        \fill[vertex2, opacity=0.6] (0,0) -- (194:3) arc (194:270:3) -- cycle;
		\node[vertex2, opacity=0.6] at (-2.5, -2.5) {$ C_2 $};
        \fill[vertex3, opacity=0.6] (0,0) -- (72:3) arc (72:194:3) -- cycle;
		\node[vertex3, opacity=0.6] at (-2.5, 2.5) {$ C_3 $};

		\drawVector{3}{1}{3}{thick};
		\drawVector{1}{3}{3}{thick};
		\drawVector{0}{-1}{3}{thick};
		\drawVector{-1}{-0.25}{2.9}{thick};
	\end{tikzpicture}
}
\hfill
\subfloat[eee]{
		\centering
		\begin{tikzpicture}[scale=0.4]
	\foreach \x in {-3,-2,-1,0,1,2,3}
		\foreach \y in {-3,-2,-1,0,1,2,3}
		{
			\fill[black!30] (\x,\y) circle (0.05);
		}
        % Axes
        \draw[->, opacity=0.4] (-3,0) -- (3,0) node[right] {$x$};
        \draw[->, opacity=0.4] (0,-3) -- (0,3) node[above] {$y$};

		\fill[vertex1, opacity=0.6] (0,0) -- (-90:3) arc (-90:18:3) -- cycle;
		\node[vertex2, opacity=0.6] at (2.5, -2.5) {$ C_1 $};
		\fill[vertex4, opacity=0.6] (0,0) -- (18:3) arc (18:72:3) -- cycle;
		\node[vertex4, opacity=0.6] at (2.5, 2.5) {$ C_4 $};
        \fill[vertex2, opacity=0.6] (0,0) -- (194:3) arc (194:270:3) -- cycle;
		\node[vertex2, opacity=0.6] at (-2.5, -2.5) {$ C_2 $};
        \fill[vertex3, opacity=0.6] (0,0) -- (72:3) arc (72:194:3) -- cycle;
		\node[vertex3, opacity=0.6] at (-2.5, 2.5) {$ C_3 $};

		\drawVector{3}{1}{3}{thick, dashed};
		\drawVector{1}{3}{3}{thick, dashed};
		\drawVector{0}{-1}{3}{thick, dashed};
		\drawVector{-1}{-0.25}{3}{thick, dashed};

		\node at (-1,1.5) {$ T_{3,1} $};
		\node at (1.5,1.5) {$ T_{4,1} $};
		\node at (1.5,-1.5) {$ T_{1,1} $};
		\node at (-1.5,-1.5) {$ T_{2,1} $};

	\end{tikzpicture}
}\caption{}\label{ex:decompo_full}
\end{figure}

\begin{figure}
	\centering
	\subfloat[ee]{
		\centering
		\begin{tikzpicture}[scale=0.4]
	\foreach \x in {-3,-2,-1,0,1,2,3}
		\foreach \y in {-3,-2,-1,0,1,2,3}
		{
			\fill[black!30] (\x,\y) circle (0.05);
		}
        % Axes
        \draw[->, opacity=0.4] (-3,0) -- (3,0) node[right] {$x$};
        \draw[->, opacity=0.4] (0,-3) -- (0,3) node[above] {$y$};

        \coordinate (A) at (-2,2);
        \coordinate (B) at (1,2);
        \coordinate (C) at (2,-2);
        \coordinate (D) at (-1,-1);
        
        % Arêtes du polytope
        \draw[thick] (A) -- (B) -- (C) -- (D) -- cycle;
        
        % Points colorés
        \fill[vertex1, opacity=0.6] (A) circle (0.2) node[above left] {$r_1$};
        \fill[vertex2, opacity=0.6] (B) circle (0.2) node[above right] {$r_2$};
        \fill[vertex3, opacity=0.6] (C) circle (0.2) node[below right] {$r_3$};
        \fill[vertex4, opacity=0.6] (D) circle (0.2) node[below left] {$r_4$};
	\end{tikzpicture}
}
	\hfill
	\subfloat[eee]{
		\centering
		\begin{tikzpicture}[scale=0.4]
    \definecolor{vertex2}{RGB}{255,0,0}
	% Grille de points entiers
	\foreach \x in {-3,-2,-1,0,1,2,3}
		\foreach \y in {-3,-2,-1,0,1,2,3}
		{
			\fill[black!30] (\x,\y) circle (0.05);
		}
        % Axes
        \draw[->, opacity=0.4] (-3,0) -- (3,0) node[right] {$x$};
        \draw[->, opacity=0.4] (0,-3) -- (0,3) node[above] {$y$};

		\fill[vertex1, opacity=0.6] (0,0) -- (-90:3) arc (-90:18:3) -- cycle;
		\node[vertex1, opacity=0.6] at (2.5, -2.5) {$ C_1 $};
		\fill[vertex4, opacity=0.6] (0,0) -- (18:3) arc (18:72:3) -- cycle;
		\node[vertex4, opacity=0.6] at (2.5, 2.5) {$ C_4 $};
        \fill[vertex2, opacity=0.6] (0,0) -- (194:3) arc (194:270:3) -- cycle;
		\node[vertex2, opacity=0.6] at (-2.5, -2.5) {$ C_2 $};
        \fill[vertex3, opacity=0.6] (0,0) -- (72:3) arc (72:194:3) -- cycle;
		\node[vertex3, opacity=0.6] at (-2.5, 2.5) {$ C_3 $};

		\drawVector{3}{1}{3}{thick};
		\drawVector{1}{3}{3}{thick};
		\drawVector{0}{-1}{3}{thick};
		\drawVector{-1}{-0.25}{2.9}{thick};
	\end{tikzpicture}
}
\hfill
\subfloat[eee]{
		\centering
		\begin{tikzpicture}[scale=0.4]
	\foreach \x in {-3,-2,-1,0,1,2,3}
		\foreach \y in {-3,-2,-1,0,1,2,3}
		{
			\fill[black!30] (\x,\y) circle (0.05);
		}
        % Axes
        \draw[->, opacity=0.4] (-3,0) -- (3,0) node[right] {$x$};
        \draw[->, opacity=0.4] (0,-3) -- (0,3) node[above] {$y$};

		\fill[vertex1, opacity=0.6] (0,0) -- (-90:3) arc (-90:18:3) -- cycle;
		\node[vertex1, opacity=0.6] at (2.5, -2.5) {$ C_1 $};
		\fill[vertex4, opacity=0.6] (0,0) -- (18:3) arc (18:72:3) -- cycle;
		\node[vertex4, opacity=0.6] at (2.5, 2.5) {$ C_4 $};
        \fill[vertex2, opacity=0.6] (0,0) -- (194:3) arc (194:270:3) -- cycle;
		\node[vertex2, opacity=0.6] at (-2.5, -2.5) {$ C_2 $};
        \fill[vertex3, opacity=0.6] (0,0) -- (72:3) arc (72:194:3) -- cycle;
		\node[vertex3, opacity=0.6] at (-2.5, 2.5) {$ C_3 $};

		\drawVector{3}{1}{3}{thick, dashed};
		\drawVector{1}{3}{3}{thick, dashed};
		\drawVector{0}{-1}{3}{thick, dashed};
		\drawVector{-1}{-0.25}{3}{thick, dashed};
		\drawVector{-1}{1}{3}{thick, dashed};
		\drawVector{2}{-1}{3}{thick, dashed};

		\node at (-2,0.5) {$ T_{3,1} $};
		\node at (-0.5,2) {$ T_{3,2} $};
		\node at (1.5,1.5) {$ T_{4,1} $};
		\node at (1,-1.5) {$ T_{1,1} $};
		\node at (2,-0.2) {$ T_{1,2} $};
		\node at (-1.5,-1.5) {$ T_{2,1} $};

	\end{tikzpicture}
}
\caption{}\label{ex:decompo_full_bad}
\end{figure}

Let $ \le_{\omega} $ be a generalized order with respect to the conic decompositon $ \{ T_{i,j} ,\  i \in I_P, j \in D_i \} $.
\begin{definition}
\label{def:new_order}
We define a preorder $\le_{P}$ on $\MonoidOfTerms$ by:
\[ a\X^u \le_{P} b\X^v \iff  \begin{cases}
		\valP(b\X^v) < \val_P(a\X^u) &(1)\\
		\textnormal{\textbf{or}} \\
		\textnormal{equality in (1)} \textnormal{\textbf{ and }} \min(I_P(b\X^v)) < \min(I_P(a\X^u)  &(2) \\
		\textnormal{\textbf{or}} \\ 
		\textnormal{equality in (2)} \textnormal{\textbf{ and }} \X^v \ge_{\omega} \X^u 
		
	\end{cases} \]

For $ f \in \LaurentPolytopalRing $, $ i \in I_P $, $ j \in D_i $ we define:
\begin{enumerate}
	\item  $\lm(f), \lc(f) $ and $ \lt(f) $, the leading monomial,
coefficient and term of $ f$ for this preorder. 
	\item $ \textnormal{in}_{P}(f) := 
		   \textnormal{in}_{r_k}(f) $ where $ k = \min(I_P(\lm(f))) $.

	   \item \nopagebreak[4] $ \lm_{i,j}(f), \lc_{i,j}(f)$ and $\lt_{i,j}(f) $ the
		   leading monomial, coefficient and term of $ \textnormal{in}_{r_i}(f) $ for the cone $ T_{i,j} $
		   in the sense of Definition \ref{def:leadings}.

	\item $ T_{i,j}(f) := \{ t \in T ,\ \lm(tf) \in T_{i,j}\cap V_{i,<}\} $
\end{enumerate}
\end{definition}

\begin{remark}
\label{rem:convergence_new}
Lemma \ref{lemma:topological_order} remains true for $ \le_P $ since 
the additional tie-breaking step happens after comparing the $ \val_P $'s.
\end{remark}

The following Lemma will be used in the proof of Proposition \ref{prop:lmij_unicity}  below.
\begin{lemma}
\label{lem:tij_module}
	Suppose $ f \in \LaurentPolytopalRing $ satisfies $ \lm(f) \in V_{i,<} $ and let $ t \in V_i $.
	Then $ \lm(tf) \in V_{_i,<} $.
\end{lemma}

\begin{proof}
\label{pr:tij_module}

	We have:
	\begin{align}
		\label{eq:proof_tij_module}
	\val_P(tf) = \min_{j \in I_P}(\val_{r_{j}}(tf))
					= \min_{j \in I_P}(\val_{r_{j}}(t) + \val_{r_j}(f)).
	\end{align}

	Since $t \in V_i$, we have $i \in I_P(t)$, that is $ \min_{j \in I_P}(\val_{r_j}(t)) = \val_{r_i}(t) $.
	Similarly, $\lm(f) \in V_{i,<} \subseteq  V_i $, so $i \in I_P(f)$ and $ \min_{j \in I_p}(\val_{r_j}(f)) = \val_{r_i}(f) $.  
	Thus for any $ (j,k) \in I_P^2 $, we have $\val_{r_i}(t) + \val_{r_i}(f) \le \val_{r_j}(t) + \val_{r_k}(f)$.
	This is true in particular when $j = k$, so the minimum in (\ref{eq:proof_tij_module}) is reached at $ i$, that is $i \in I_P(tf)$.  

	Now we show that if $ j \in I_P(tf) $, then $ i \le j $.
	Suppose that $ j \in I_P(tf) $.
	Then $\val_{r_j}(t) + \val_{r_j}(f) = \val_{r_i}(t) + \val_{r_i}(f)$.
	Since $t \in V_i$, $\val_{r_j}(t) \geq \val_{r_i}(t)$,
	so $\val_{r_j}(f) \leq \val_{r_i}(f)=\val_P(f).$
	Thus $\val_{r_j}(f)=\val_P(f)$ and $ j \in I_P(f) $. 
%	We have $ j \in I_P(f) $, otherwise if $ j \notin I_P(f) $, we get:
%\[ \val_{r_j}(t) + \val_{r_j}(f) > \val_{r_j}(t) + \val_{r_i}(f) \ge \val_{r_i}(t) + \val_{r_i}(f), \]
%	and so $ j \notin I_P(tf) $, contradiction.
	In addition $ \lm(f) \in V_{i,<} $, thus $ i = \min(I_P(f))$, and then $ i \le j $.

	Finally $ i = \min(I_P(tf)) $, that is $ \lm(tf) \in V_{i,<} $.
\end{proof}

\begin{proposition}
\label{prop:lmij_unicity}
Let $f, g \in \LaurentPolytopalRing$ and $a \in T_{i,j}(f)$. Then:
\begin{enumerate}
    \item $\lm(af) = a \lm_{i,j}(f)$.
    \item $ \lm_{i,j}(f) T_{i,j}(f) \cap \lm_{i,j}(g) T_{i,j}(g) \subseteq T_{i,j} \cap V_{i,<}$ is a finitely generated $T_{i,j}$-module.
\end{enumerate}
\end{proposition}

\begin{proof}
\label{proof:lmij_unicity}
\begin{enumerate}
	\item By definition, $a \in T_{i,j}(f)$ means that $\lm(af) \in T_{i,j} \cap V_{i,<} \subseteq V_{i,<}$.
The elements of $V_{i,<}$ are precisely the monomials $t$ in $T$ such that $\min(I_P(t)) = i$, so
$\min(I_P(\lm(af))) = i$.  
By ($2$) of Definition \ref{def:new_order}, it follows that $\textnormal{in}_{P}(af) = \textnormal{in}_{r_i}(af)$, 
and the latter is equal to $a \times \textnormal{in}_{r_i}(f)$, since $\val_{r_i}$, unlike $\val_P$,
satisfies $\val_{r_i}(af) = \val_{r_i}(a) + \val_{r_i}(f)$.  

Now by Definition \ref{def:new_order}, $\lm(af)$ is the greatest term under the generalized order
$\le_{\omega}$ of $\textnormal{in}_{P}(af) = a \times \textnormal{in}_{r_i}(f)$.  
Since $a \in  T_{i,j}(f)$, this greatest term is equal to
$a \times \lm_{i,j}(f)$ by (3) of Definition \ref{def:new_order}.
\item Let $A(i,j,f,g) := \lm_{i,j}(f) T_{i,j}(f) \cap \lm_{i,j}(g) T_{i,j}(g)$.  
	We first show that $ A(i,j,f,g) $ is a $ T_{i,j} $-module. 
	For this, it is sufficient to show that if $h \in \LaurentPolytopalRing$, then $T_{i,j}(h)$ is a $T_{i,j}$-module.  
	Indeed, if this holds, then $\lm_{i,j}(f)T_{i,j}(f)$ and $\lm_{i,j}(g)T_{i,j}(g)$ are also $T_{i,j}$-modules,  
	and so is their intersection $ A(i,j,f,g) $.

	Let $ a \in T_{i,j}(h)$ and $ t \in T_{i,j} $. We aim to show
	that $ at \in T_{i,j}(h) $, that is $ \lm(tah) \in T_{i,j}\cap V_{i,<} $.

	% Since $ t \in T_{i,j} \subseteq V_i $ and $ \lm(ah) \in V_{i,<} $, we can apply Lemma \ref{lem:tij_module} to $ t $
	% and $ ah $ to get $ \lm(tah) \in V_{i,<} $. We then get by item ($ 2 $) of Definition \ref{def:new_order}
	% $ \textnormal{in}_P(tah) = \textnormal{in}_{r_i}(tah) = ta\times \textnormal{in}_{r_i}(h) $,
	% and $  \lm(taf) $ equals the maximum under $ \le_{\omega} $ of the terms of $ t\times\a \textnormal{in}_{r_i}(h) $.
	% But $ a \in T_{i,j}(h) $, so the greatest term of $ a \times \textnormal{in}_{r_i(f)} $ is $ a \times \lm_{i,j}(f) $
	% and is in $ T_{i,j} $. Now $ t $ is in $ T_{i,j} $, so by item 2. of Definition \ref{def:gmo}, $ \lm(taf) \in T_{i,j} $ also.
	% This conclude the proof that $ A(i,j,f,g) $ is a $ T_{i,j} $-module.

Since $ t \in T_{i,j} \subseteq V_i $ and $ \lm(ah) \in V_{i,<} $, 
we can apply Lemma \ref{lem:tij_module} to $ t $ and $ ah $, 
yielding $ \lm(tah) \in V_{i,<} $. 
By item (2) of Definition \ref{def:new_order}, we then have:
\[
\textnormal{in}_P(tah) = \textnormal{in}_{r_i}(tah) = ta \times \textnormal{in}_{r_i}(h).
\]
So $ \lm(tah) $ equals the maximum, under $ \le_{\omega} $, of the terms of $ t \times a \times \textnormal{in}_{r_i}(h) $. 
Since $ a \in T_{i,j}(h) $, the greatest term of $ a \times \textnormal{in}_{r_i}(h) $ 
is $ a \times \lm_{i,j}(h) $, and this term belongs to $ T_{i,j} $.

Now, since $ t \in T_{i,j} $, by item (2) of Definition \ref{def:gmo}, 
the maximum on the terms of $ t \times a \times \textnormal{in}_{r_i}(h) $, which equals $ \lm(tah) $, 
is in $ T_{i,j} $.

% Since \( t \in T_{i,j} \subseteq V_i \) and \( \lm(ah) \in V_{i,<} \),
% we can apply Lemma \ref{lem:tij_module} to \( t \) and \( ah \), 
% yielding \( \lm(tah) \in V_{i,<} \).
% By item (2) of Definition \ref{def:new_order}, we then have:
%
% \[
% \textnormal{in}_P(tah) = \textnormal{in}_{r_i}(tah) = ta \times \textnormal{in}_{r_i}(h).
% \]
%
% So \( \lm(tah) \) equals the maximum, under \( \le_{\omega} \), of the terms of \( t \times a \times \textnormal{in}_{r_i}(h) \).
% Since \( a \in T_{i,j}(h) \), the greatest term of \( a \times \textnormal{in}_{r_i}(h) \)
% is \( a \times \lm_{i,j}(h) \), and this term belongs to \( T_{i,j} \).
%
% Now, since \( t \in  T_{i,j} \), by item (2) of Definition \ref{def:gmo},
% the maximum on the terms of $ t \times a \times \textnormal{in}_{r_i}(h) $, whih equals $ \lm(ath) $,
% is in $ T_{i,j} $.
So $ \lm(tah) \in V_{i,<}\cap T_{i,j} $, and 
this completes the proof that \( A(i,j,f,g) \) is a \( T_{i,j} \)-module.

	Now we prove that any subset $X \subseteq T_{i,j}$ is finitely generated as a $T_{i,j}$-module.  
This is equivalent to showing that the following partial order on $X$:
\[
a \preceq b \iff \exists u \in T_{i,j},\ b = au,
\]
is a well partial order.
Let $v_1, \dots, v_k$ be a finite generating set of $T_{i,j}$. For every $t \in X \subseteq T_{i,j}$, 
there exists non-negative integers $ n_1, \dots, n_k $ (depending on $ t $) such that $t = v_{1}^{n_1} \dots v_{k}^{n_k}$.
We embed $X$ into $\N^k$ via the map $\phi: t = v_{1}^{n_1} \dots v_{k}^{n_k} \mapsto (n_1, \dots, n_k) \in \N^k.$
Under this embedding, the partial order $\preceq$ on $X$ corresponds to the component-wise order $\leq$ on $\phi(X)$ in $\N^k$:  
\[
a \preceq b \text{ in } X \iff \phi(a) \leq \phi(b) \text{ in } \N^k,
\]
so we have that $\preceq$ is a well partial order on $X$ if and only if $\leq$ is a well partial order on $\phi(X)$.  
By Dickson's Lemma, the component-wise order on any subset of $\N^k$ is a well partial order, so $\preceq$ is a well partial order on $X$.

Applying this result to $X = A(i,j,f,g)$, we conclude that $A(i,j,f,g)$ is finitely generated as a $T_{i,j}$-module.
\end{enumerate}
\end{proof}

With these adapted definitions, the definition of Gröbner basis becomes:

\begin{definition}
    Let $J$ be an ideal in $\LaurentPolytopalRing$ and $G$ be a finite subset of $J \setminus \{0\}$.  
    We say that $G$ is a Gröbner basis of $J$ (with respect to the g.m.o $\le_{\omega}$ for the conic decomposition $\{T_{i,j}, i \in I_P, j \in D_i\}$) when:
    \[
    \lm(J) = \bigcup_{g \in G, i \in I_P, j \in D_i} T_{i,j}(g) \lm_{i,j}(g).
    \]
\end{definition}

\subsection{Buchberger algorithm in $ \LaurentPolytopalRing $ }

We now adapt the multivariate division algorithm (Algorithm \ref{alg:multi_div}),
the Buchberger criterion (Proposition \ref{prop:buch_criterion}),
and the Buchberger algorithm (Algorithm \ref{alg:buchberger}) from Sections 
\ref{sec:MultiVarDivision} and \ref{sec:GBpolytopal} to our new setting.

This process involves substituting \( \lm_i(f) \) and \( T_i(f) \) with their
newly defined counterparts, \( \lm_{i,j}(f) \) and \( T_{i,j}(f) \), respectively, and substituing $ \le_r $
with $ \le_P $.
Any additional changes beyond these substitutions are also noted where applicable. 
Full proofs can be found in the Appendices.

To simplify notation when indexing the cones in the conic decomposition
\( \{T_{i,j},\ i \in I_P,\ j \in D_i\} \), we define \( L := \{(i,j) \mid i \in I_P,\ j \in D_i \} \).

\subsubsection{Division algorithm}

In Proposition \ref{prop:multi_div_new}, the only difference from 
Proposition \ref{prop:multi_div} is the replacement of $ \le_r $ with $ \le_P $ in (3).

\begin{proposition}
	\label{prop:multi_div_new}
	Let $f \in \LaurentPolytopalRing$ and $G$ be a finite subset of $\LaurentPolytopalRing$. 
	Algorithm \ref{alg:multi_div} produces a family $(q_g)_{g \in  G}$ and $r$ in $\LaurentPolytopalRing$ such that:
	\begin{enumerate}[(1)]
		\item $f = \sum_{g \in  G}q_gg + r$
		\item for all monomial $t$ in $r$,
			\begin{equation}
				t \notin \bigcup_{(i,j) \in L, g \in G}T_{i,j}(g)\lm_{i,j}(g)
			\end{equation}
		\item for all $g \in  G$ and all monomial $t$ in $q_g$, $\lt(tg) \le_{P} \lt(f)$.
	\end{enumerate}
\end{proposition}

% \begin{proof}
% We construct by induction sequences $(f_j)_{j\ge 0}$, $(q_{g,j})_{j \ge 0}$ for $g \in G$ and $(r_j)_{j\ge 0}$ such that for all $j \ge 0$:
% \[ f = f_j + \sum_{g \in G}q_{g,j}g + r_j,\]
% and $\lt(f_j)_{j\ge 0}$ is strictly decreasing. We first set $f_0 = f$, $r_0 = 0$ and $q_{g,0} = 0$ for all $g \in G$. If there exists $i \in I$ and $g \in G$ such that \[\lm \left( \frac{\lm(f_j)}{\lm_i(g)}g \right) = \lm(f_j) \in T_i, \] we set $f_{j+1} = f_j - tg$ and $q_{g,j+1} = q_{g,j} + tg$ where $t = \frac{\lt(f_j)}{\lt_i(g)}$, and leave unchanged $r_j$ and the other $q_{g,j}$'s. 
% Otherwise, we set $f_{j+1} = f_j - lt(f_j)$ and $r_{j+1} = r_j + lt(f_j)$ and leave unchanged the $q_{g,j}$'s. 
% By construction, the sequence $(\lt(f_j))_{j \ge 0}$ is strictly decreasing.
% By Lemma \ref{lemma:topological_order}, we deduce that 
% $\valP(r_{j+1}-r_j)$ and the $\valP(q_{g,j+1}-q_{g,j})$'s tend to $+\infty$ when $j \to +\infty$.
% Thus $r_j$ and the $q_{g,j}$'s converge in $K\{\mathbf{X};P\}$.
% Their limits % elements $q_g=\lim_{j \to +\infty}q_{g,j}$ for $g \in G$ and $r = \lim_{j \to +\infty}r_j$ 
% satisfy the requirements of the proposition.
% \end{proof}

\begin{algorithm}
	\SetKwInOut{Input}{input}\SetKwInOut{Output}{output}
	\caption{Multivariate division algorithm in $K\{\mathbf{X};P\}$} 
	\label{alg:multi_div_new}
	\Input{$f,g_1,\dots,g_m \in K\{\mathbf{X};P\}$}
	\Output{$q_1,\dots,q_m,r$ satisfying Prop \ref{prop:multi_div_new}}
	\BlankLine
	$q_1,\dots,q_m,r\gets 0$\;
	\While{ $f \neq 0$}{
		%$i \gets$ index in $I$ such that $\textnormal{lm}(f) \in T_i$\;
		\While{ $\exists (i,j) \in L$ and $k \in  \llbracket 1,m \rrbracket$ such that $\lm\left(\frac{\lm(f)}{\lm_{i,j}(g_k)}g_k\right) = \lm(f)$}{
			$t \gets \frac{\lt(f)}{\lt_{i,j}(g_k)}$\;
			$q_k \gets q_k + t $\;
			$f \gets f - tg_k $\;
		}
		$r \gets r + \lt(f)$

		$f \gets f - \lt(f)$;
	}
	\Return $q_1,\dots,q_m,r$
\end{algorithm}

\subsubsection{S-pairs and Buchberger criterion}

% The Definition \ref{def:Spair_new} of $S$-pair  is exactly the Same as Definition $\ref{def:Spair}  $.
% In Lemma \ref{lemma:sumSpair_new}, the monomial $ u $
% in which the colliding occurs needs to be in $ T_{i,j}\cap V_{i,<} $ (instead of just $ T_{i,j} $ in Lemma \ref{lemma:sumSpair})
% to accomodate the new definition of $ T_{i,j}(f) $ ((4) of Definition \ref{def:new_order}).
% In the proof of Proposition \ref{prop:buch_criterion_new} (Appendix \ref{app:proof_criterion}),
% the Lemma \ref{lemma:inequality} is replaced by a refined version \ref{lemma:inequality_new}
% taking into account the intermediate tie-breaking involving the indices $ I_P $.
% Otherwise, the statments are the same.

The definition of \( S \)-pair in Definition \ref{def:Spair_new} is identical to that in Definition \ref{def:Spair}. 
In Lemma \ref{lemma:sumSpair_new}, the monomial \( u \) where the collision occurs must belong to \( T_{i,j} \cap V_{i,<} \) (instead of just \( T_{i,j} \) as in Lemma \ref{lemma:sumSpair}) to accommodate the new definition of \( T_{i,j}(f) \) (item (4) of Definition \ref{def:new_order}). 

In the proof of the Buchberger criterion of Proposition \ref{prop:buch_criterion_new} (Appendix \ref{app:proof_criterion}), Lemma \ref{lemma:inequality} is replaced by a refined version, Lemma \ref{lemma:inequality_new}, which accounts for the intermediate tie-breaking involving the indices \( I_P \). 

Otherwise, the statements remain the same.

\begin{definition}[S-pair]
	\label{def:Spair_new}
	Let $f,g \in \LaurentPolytopalRing$ and $(i,j) \in L$. 
	For $v \in \lm_{i,j}(f)T_{i,j(}(f) \cap \lm_{i,j}(g)T_{i,j}(g)$, we define:
	\[ S((i,j),f,g,v) := \lc_{i,j}(g)\frac{v}{\lm_{i,j}(f)}f - \lc_{i,j}(f)\frac{v}{\lm_{i,j}(g)}g.\]
\end{definition}

\begin{lemma}
	\label{lemma:sumSpair_new}
Let $h_1,\dots,h_m \in \LaurentPolytopalRing$ and $(i,j) \in L$. 
For $1\le k \le m-1$, let $U((i,j),h_k,h_{k+1})$ be a finite system of generators of  $\lm_{i,j}(h_k)T_{i,j}(h_k) 
\cap \lm_{i,j}(h_{k+1})T_i(h_{k+1})$ which exists by Proposition \ref{prop:lmij_unicity}.
%\TVshort{J'ai changé la référence}
%\ref{prop:Uifg}.
Suppose that there are $t_1,\dots, t_m \in \MonoidOfTerms$, $u \in T_{i,j} \cap V_{i,<}$ and $c \in \val{}(K^\times)$ such that
\begin{itemize}
	\item for all $k \in \{1,\dots,m\}$, $\lt(t_kh_k) = c_ku$ with $\textnormal{val}(c_k) = c$
	\item  $\lt(\sum_{k=1}^{m}t_kh_k) <_P c_1u$.
\end{itemize}

Then there are elements $d_k \in K$, $v_k \in U((i,j),h_k,h_{k+1})$ for $1 \le k \le m-1$ and $t_m^\prime \in \LaurentPolytopalRing$ such that:
\begin{enumerate}[(1)]
	\item $\sum_{j=k}^{m}t_kh_k = \sum_{k=1}^{m-1}d_j\frac{u}{v_k}S((i,j),h_k,h_{k+1},v_k) + t_{m}^\prime h_m$. \label{enum:1_new}
	\item $\valP(t^\prime_mh_m) > \valP(uc_1)$. \label{enum:2_new}
	\item $\frac{u}{v_k} \in T_{i,j}$ for all $k <m$. \label{enum:3_new}
	\item For all $k <m$, $\val \left( d_k \lc_{i,j}(h_k) \lc_{i,j}(h_{k+1}) \right) \geq c.$ \label{enum:4_new}
	% \item $\val{P}(d_j\frac{u}{r_j}) + \max(\val{P}(h_j),\val{P}(h_{j+1})) \ge \val{P}(uc_1)$ for $j < m$.
\end{enumerate}

\end{lemma}

\begin{proposition}[Buchberger criterion]
	\label{prop:buch_criterion_new}
Let $H = (h_1,\dots,h_m)$ be a family in $\LaurentPolytopalRing$ and $J$ the ideal generated by $H$.
For each $(i,j) \in I$ and $h_k \neq h_s \in H$, let $U((i,j),h_k,h_s)$ be a finite system of generators of 
the $T_{i,j}$-module $\lm_{i,j}(h_s)T_{i,j}(h_s) \cap \lm_{i,j}(h_k)T_{i,j}(h_k)$.
The following are equivalent:
\begin{enumerate}[(1)]
	\item $H$ is a Gröbner basis of $J$
	\item For all $(i,j)  \in L$, $h_k \neq h_s$, $v \in U((i,j),h_s,h_k)$:
		\begin{equation}
			\rem(S((i,j),h_s,h_k,v),H) = 0
		\end{equation}
		.
	\end{enumerate}
\end{proposition}
\subsubsection{Buchberger algorithm}

In the Buchberger algorithm (Algorithm \ref{alg:buchberger_new}), we just update line $ 7 $ to use
Algorithm \ref{alg:multi_div_new} when reducing an $ S $-pair.
% \begin{proposition}
% 	\label{prop:buchberger_new}
% Algorithm \ref{alg:buchberger_new} on page \pageref{alg:buchberger_new} is correct and terminates, in the sense that it 
% calls the multivariate $\textrm{division}$ a finite number of times.
% \end{proposition}
\begin{algorithm}
	\SetKwInOut{Input}{input}\SetKwInOut{Output}{output}
	\caption{Buchberger algorithm in $\LaurentPolytopalRing$}
	\label{alg:buchberger_new}
	\Input{$J = (h_1,\dots,h_m)$ an ideal of $\LaurentPolytopalRing$}
	\Output{a Gröbner basis of $J$}
	\BlankLine
	$H \gets \{h_1,\dots,h_m\}$; $B \gets \{(h_s,h_k),\ 1 \le s < k\le m$\}

	\While{ $ B \neq \emptyset $}{
		$(f,g) \gets$ element of $B$; $B \gets B \setminus \{(f,g)\}$\;
		\For{ $ (i,j) \in L$ }{
			$U((i,j),f,g) \gets$ finite set of generators of $\lm_{i,j}(f)T_{i,j}(f)\cap  \lm_{i,j}(g)T_{i,j}(g)$\; 
			\For{ $v \in U((i,j),f,g)$}{
				$\_,r \gets \textnormal{division}(S((i,j),f,g,v),H)$  \tcp*{Algo \ref{alg:multi_div_new}}
				\If{ $r \neq 0$}{
					$B \gets B \cup \{(h,r),\  h \in H\}$\;
					$H \gets H \cup \{r\}$
				}
			}
		}
	}
	\Return $H$
\end{algorithm}

% \subsection{Computing $ T_{i,j}(f) $}

% \subsection{Generalistion to Polyhedral Algebra}

\bibliographystyle{elsarticle-num}
\bibliography{biblio.bib}

\begin{thebibliography}{10}
\expandafter\ifx\csname url\endcsname\relax
  \def\url#1{\texttt{#1}}\fi
\expandafter\ifx\csname urlprefix\endcsname\relax\def\urlprefix{URL }\fi
\expandafter\ifx\csname href\endcsname\relax
  \def\href#1#2{#2} \def\path#1{#1}\fi

\bibitem{EKL:2006}
M.~Einsiedler, M.~Kapranov, D.~Lind,
  \href{https://doi.org/10.1515/CRELLE.2006.097}{Non-archimedean amoebas and
  tropical varieties}, Journal für die reine und angewandte Mathematik
  2006~(601) (2006) 139--157.
\newline\urlprefix\url{https://doi.org/10.1515/CRELLE.2006.097}

\bibitem{CVV:2019}
X.~Caruso, T.~Vaccon, T.~Verron,
  \href{https://hal.science/hal-01995881}{{Gr{\"o}bner bases over Tate
  algebras}}, in: {ISSAC 2019 - International Symposium on Symbolic and
  Algebraic Computation}, Beijing, China, 2019.
\newline\urlprefix\url{https://hal.science/hal-01995881}

\bibitem{PU:1999}
F.~Pauer, A.~Unterkircher, Gr{\"o}bner bases for ideals in laurent polynomial
  rings and their application to systems of difference equations, Applicable
  Algebra in Engineering, Communication and Computing 9 (1999) 271--291.

\bibitem{Tate:1971}
J.~Tate, Rigid analytic spaces, Inventiones mathematicae 12~(4) (1971)
  257--289.

\bibitem{Raynaud:1994}
M.~Raynaud, Rev{\^e}tements de la droite affine en caract{\'e}ristique p> 0 et
  conjecture d'abhyankar, Inventiones mathematicae 116 (1994) 425--462.

\bibitem{Gubler:2007}
W.~Gubler, Tropical varieties for non-archimedean analytic spaces, Inventiones
  mathematicae 169~(2) (2007) 321--376.

\bibitem{Rabinoff:2012}
J.~Rabinoff, Tropical analytic geometry, newton polygons, and tropical
  intersections, Advances in Mathematics 229~(6) (2012) 3192--3255.

\bibitem{FM:2023bis}
N.~Friedenberg, K.~Mincheva, Tropical adic spaces i: The continuous spectrum of
  a topological semiring (2023).
\newblock \href {http://arxiv.org/abs/2209.15116} {\path{arXiv:2209.15116}}.

\bibitem{Gubler:2007bis}
W.~Gubler, The bogomolov conjecture for totally degenerate abelian varieties,
  Inventiones mathematicae 169~(2) (2007) 377--400.

\bibitem{CVV:2020}
X.~Caruso, T.~Vaccon, T.~Verron, Signature-based algorithms for gr\"{o}bner
  bases over tate algebras, ISSAC '20, Association for Computing Machinery, New
  York, NY, USA, 2020.

\bibitem{CVV:2021}
X.~Caruso, T.~Vaccon, T.~Verron, On fglm algorithms with tate algebras, in:
  Proceedings of the 2021 on International Symposium on Symbolic and Algebraic
  Computation, ISSAC '21, Association for Computing Machinery, New York, NY,
  USA, 2021, p. 67–74.

\bibitem{CVV:2022}
X.~Caruso, T.~Vaccon, T.~Verron, On polynomial ideals and overconvergence in
  tate algebras, in: Proceedings of the 2022 International Symposium on
  Symbolic and Algebraic Computation, ISSAC '22, Association for Computing
  Machinery, New York, NY, USA, 2022, p. 489–497.

\bibitem{VV:2023}
T.~Vaccon, T.~Verron, Universal analytic gr\"{o}bner bases and tropical
  geometry, ISSAC '23, Association for Computing Machinery, New York, NY, USA,
  2023, p. 517–525.

\bibitem{sagemath}
{The Sage Developers}, {S}ageMath, the {S}age {M}athematics {S}oftware {S}ystem
  ({V}ersion 10.3), {\tt https://www.sagemath.org} (2024).

\bibitem{BLV:2024}
M.~Barkatou, L.~Legrand, T.~Vaccon, Gr{\"o}bner bases over polytopal affinoid
  algebras, in: Proceedings of the 2024 International Symposium on Symbolic and
  Algebraic Computation, 2024, pp. 188--197.

\end{thebibliography}
\newpage
\begin{appendices}

\section{Appendix: Sagemath implementation} 
\label{annex:sage}
We have implemented in \sage all the algorithms presented in this paper for the Laurent polynomial case: Algorithms \ref{alg:multi_div} and \ref{alg:buchberger} can be applied over $\LaurentRing$ as in \cite{PU:1999}, the main difference being the way leading terms are defined. As far as we know, this is the first ever implementation of the ideas
of \cite{PU:1999}. It has been made possible thanks to Section \ref{sec:implem}.
The integration of our implementation into the \sage sources is currently underway. Discussions, comments and progress can be tracked on the official GitHub repository at:
\begin{center}
{https://github.com/sagemath/sage/pull/37241}
\end{center}

The case of polytopal affinoid algebras, which is predominantly built upon the polynomial case, is currently a work in progress.
% \begin{center}
%     https://github.com/vilanele/sage/tree/polytopal\_algebras
% \end{center}
We plan to integrate it as well into the \sage sources as soon as it is ready.
% While waiting for integration, the code can be run locally by fetching the 'buchberger\_laurent' branch from the repository 
% \begin{center}
% https://github.com/vilanele/sage 
% \end{center}
%  into a local Sage installation. 
% Note that, as a Cython file is modified in the branch, a Sage rebuild is necessary.

\bigskip
\textbf{Short demo}. A new class constructor named {\color{constructor}\verb?GeneralizedOrder?} is introduced.
The supported g.m.o's are those for which the underlying conic decomposition is the decomposition of Example \ref{ex:standard_conic_decomposition} and for which the order is defined by
a group order and a score function as in Lemma \ref{lemma:construct_gto}.
The group order and score function can be specified using the keywords {\color{keyword}\verb?group_order?} and {\color{keyword}\verb?score_function?} in the constructor.
Defaults are the lexicographical order on $\Z^n$ and the {\color{string}\verb?min?} function of Example \ref{ex:standard_gmo_2}.

\bigskip
{\noindent \small
\begin{tabular}{rl}
	\cIn & {\color{import}\verb?from?}\verb? sage.rings.polynomial.generalized_order ?{\color{import}\verb?import?} \\
		 & \verb?GeneralizedOrder? \\
	\cIn & {\color{constructor}\verb?GeneralizedOrder?}\verb?(3)? \\
	\cOut & \verb?Generalized order in 3 variables using (lex, min)? 
\end{tabular}}

\bigskip
Another example, using this time the score function {\color{string}\verb?degmin?} of Example \ref{ex:standard_gmo_1}:

\bigskip
{\noindent \small
\begin{tabular}{rl}
\cIn & {\color{parent}\verb?G?}\verb? = ?{\color{constructor}\verb?GeneralizedOrder?}\verb?(2, ?{\color{keyword}\verb?score_function?}\verb?='?{\color{string}\verb?degmin?}\verb?')?\verb?; ?{\color{parent}\verb?G?}  \\
	\cOut & \verb?Generalized order in 2 variables using (lex, degmin)? 
\end{tabular}}

\bigskip
Now we can compare tuples:

\bigskip
{\noindent \small
\begin{tabular}{rl}
	\cIn & {\color{parent}\verb?G?}\verb?.?{\color{method}\verb?greatest_tuple?}\verb?((-2,3), (1,2))?\\
	\cOut &	$(-2,3)$  \\
	\cIn & {\color{parent}\verb?G?}\verb?.?{\color{method}\verb?greatest_tuple_for_cone?}\verb?(2, (1,3), (-1,2), (-4,-3))?\\
	\cOut &	$(-4, -3)$
\end{tabular}}

\bigbreak
The {\color{constructor}\verb?LaurentPolynomialRing?} constructor has been updated to accept instances of the new {\color{constructor}\verb?GeneralizedOrder?} class for the keyword {\color{keyword}\verb?order?}.
Elements have new methods:

\bigskip
{\noindent \small
\begin{tabular}{rl}
	\cIn & {\color{parent}\verb?L?}\verb?.<x,y> = ?{\color{constructor}\verb?LaurentPolynomialRing?}\verb?(QQ, ?{\color{keyword}\verb?order?}\verb?=G)? \\
	\cIn & \verb?f = 2*x^2*y^-1 + x^-3*y - 3*y^-5? \\
	\cIn & \verb?f.?{\color{method}\verb?leading_monomial?}\verb?()                         //? $\lm(f)$ \\
	\cOut & $y^{-5}$ \\
	\cIn & \verb?f.?{\color{method}\verb?leadin_monomial_for_cone?}\verb?(1)                //? $\lm_1(f)$ \\
	\cOut & $x^{-3}y$ \\
	\cIn & \verb?f.?{\color{method}\verb?generator_for_cone?}\verb?(2)?\verb?                      //? $T_2(f)$ \\
	\cOut & $xy^2$
\end{tabular}}

We can reduce an element using multivariate division (Algorithm \ref{alg:multi_div}):
\bigskip

{\noindent \small
\begin{tabular}{rl}
	\cIn & \verb?L = [x^-2*y^-1 + x*y, x^-2*y + x^2*y^-1]? \\
	\cIn &\verb?f.?{\color{method}\verb?generalized_reduction?}\verb?(L)? \\
	\cOut & $(-y^3 + 2x^2y^{-1} - 3x^{-1}y^{-1}, [x^{-1}y^2 + 3x^{-2}y^{-2}, -3x^{-2}y^{-4}])$
\end{tabular}}

\bigskip
Lastly, within the {\color{constructor}\verb?LaurentPolynomialIdeal?} class, the method {\color{method}\verb?groebner_basis?} has been modified so that it uses Algorithm \ref{alg:buchberger} when a generalized order is specified:

\bigskip
{\noindent \small
\begin{tabular}{rl}
	\cIn & {\color{parent}\verb?G?}\verb? = ?{\color{constructor}\verb?GeneralizedOrder?}\verb?(3, ?{\color{keyword}\verb?score_function?}\verb?='?{\color{string}\verb?degmin?}\verb?')? \\
	\cIn & {\color{parent}\verb?L?}\verb?.<x,y,z> = ?{\color{constructor}\verb?LaurentPolynomialRing?}\verb?(QQ, ?{\color{keyword}\verb?order?}\verb?=G)? \\
	\cIn & {\color{parent}\verb?I?}\verb? = ?{\color{parent}\verb?L?}\verb?.?{\color{method}\verb?ideal?}\verb?([x^-3*y^-4 + x*y*z, x^3*y^-2 + y^-1*z])? \\
	\cIn & {\color{parent}\verb?I?}\verb?.?{\color{method}\verb?groebner_basis?}\verb?()? \\
	\cOut & $(x^3y^{-4} + xyz,$ \\
		  & $\ x^3y^{-2} + y^{-1}z,$ \\
		  & $\ -y^{-4} + x^{-1}y^{-2}z^{-1})$  \\
	% \cOut & $(x^3y^{-4} + xyz,\ x^3y^{-2} + y^{-1}z,\ -y^4z + x^{-1}y^{-2}z^{-1})$
\end{tabular}}

\bigskip
Another example of Gröbner basis computation for an ideal in the ring $\Q[x^{\pm 1}, y^{\pm 1}, z^{\pm 1}]$:
\bigskip

{\noindent \small
\begin{tabular}{rl}
	\cIn & {\color{parent}\verb?G?}\verb? = ?{\color{constructor}\verb?GeneralizedOrder?}\verb?(3, ?{\color{keyword}\verb?score_function?}\verb?='?{\color{string}\verb?min?}\verb?')? \\
	\cIn & {\color{parent}\verb?L?}\verb?.<x,y,z> = ?{\color{constructor}\verb?LaurentPolynomialRing?}\verb?(QQ, ?{\color{keyword}\verb?order?}\verb?=G)? \\
	\cIn & \verb?g1 = 1/2*x^-1*y + 3*y^-4*z^2 + y? \\
	\cIn & \verb?g2 = 2*x^2*y^3*z^-1 - 1/3*x^-1*y^3*z^-6? \\
	\cIn & {\color{parent}\verb?I?}\verb? = ?{\color{parent}\verb?L?}\verb?.?{\color{method}\verb?ideal?}\verb?([g1, g2])? \\
	\cIn & {\color{parent}\verb?I?}\verb?.?{\color{method}\verb?groebner_basis?}\verb?()? \\
	\cOut & $(y + \frac{1}{2}x^{-1}y + 3y^{-4}z^2,$ \\
		  & $\ 2x^2y^3z^{-1} - \frac{1}{3}x^{-1}y^3z^{-6},$ \\
          & $\ \frac{1}{4}y^5z^5 - 3x^2z^7 + \frac{3}{2}xz^7 + \frac{1}{3}y^5 + z^2,$ \\
		  & $\ \frac{1}{4}y^{10}z^5 - \frac{3}{4}y^5z^7 + \frac{1}{3}y^{10} - \frac{9}{2}xz^9 + 2y^5z^2 + 3z^4,$ \\
		  & $\ \frac{1}{4}y^{15}z^5 + \frac{1}{3}y^{15} + 3y^{10}z^2 + 9y^5z^4 + 9z^6,$ \\
          & $\ 6x^2y^4z^4 + 3xy^4z^4 + 3x^{-1}y^{-1}z)$
\end{tabular}}

\bigskip
Same ideal as above, but using the score function {\color{string}\verb?degmin?}:
\bigskip

{\noindent \small
\begin{tabular}{rl}
	\cIn & {\color{parent}\verb?G?}\verb? = ?{\color{constructor}\verb?GeneralizedOrder?}\verb?(3, ?{\color{keyword}\verb?score_function?}\verb?='?{\color{string}\verb?degmin?}\verb?')? \\
	\cIn & {\color{parent}\verb?L?}\verb?.<x,y,z> = ?{\color{constructor}\verb?LaurentPolynomialRing?}\verb?(QQ, ?{\color{keyword}\verb?order?}\verb?=G)? \\
	\cIn & \verb?g1 = 1/2*x^-1*y + 3*y^-4*z^2 + y? \\
	\cIn & \verb?g2 = 2*x^2*y^3*z^-1 - 1/3*x^-1*y^3*z^-6? \\
	\cIn & {\color{parent}\verb?I?}\verb? = ?{\color{parent}\verb?L?}\verb?.?{\color{method}\verb?ideal?}\verb?([g1, g2])? \\
	\cIn & {\color{parent}\verb?I?}\verb?.?{\color{method}\verb?groebner_basis?}\verb?()? \\
	\cOut & $(y + \frac{1}{2}x^{-1}y + 3y^{-4}z^2,$ \\
		  & $\ 2x^2y^3z^{-1} - \frac{1}{3}x^{-1}y^3z^{-6},$ \\
		  & $\ y^5z^3 + \frac{1}{3}x^{-2}y^5*z^{-2} + x^{-2},$ \\
		  & $\ \frac{-1}{16}y^5z^6 - \frac{1}{12}y^5z - \frac{1}{4}z^3 + \frac{1}{8}x^{-1}z^3 - \frac{1}{16}x^{-2}z^3,$ \\
		  & $\ -\frac{1}{6}xy^3z^{-1} + \frac{1}{24}x^{-1}y^3z^{-1} - \frac{1}{12}x^{-2}y^{-2}z^{-4} + \frac{1}{24}x^{-3}y^{-2}z^{-4},$ \\
		  & $\ -\frac{1}{36}y^3z^{-1} - \frac{1}{72}x^{-1}y^3z^{-1} - \frac{1}{72}x^{-3}y^{-2}z^{-4})$
\end{tabular}}

\bigskip
An example in the ring $\mathbb{F}_9[x^{\pm 1}, y^{\pm 1}]$:
\bigskip

{\noindent \small
\begin{tabular}{rl}
	\cIn & {\color{parent}\verb?G?}\verb? = ?{\color{constructor}\verb?GeneralizedOrder?}\verb?(2, ?{\color{keyword}\verb?score_function?}\verb?='?{\color{string}\verb?degmin?}\verb?')? \\
	\cIn & {\color{parent}\verb?F?}\verb? = ?{\color{constructor}\verb?GF?}\verb?(9)? \\
	\cIn & {\color{parent}\verb?L?}\verb?.<x,y> = ?{\color{constructor}\verb?LaurentPolynomialRing?}\verb?(?{\color{parent}\verb?F?},{\color{keyword}\verb? order?}\verb?=G)? \\
	\cIn & \verb?g1 = x^2*y + y^-6? \\
	\cIn & \verb?g2 = x^3*y^-2 + x^-6*y? \\
	\cIn & \verb?g3 = x^-2*y + x^-1*y^-2? \\
	\cIn & {\color{parent}\verb?I?}\verb? = ?{\color{parent}\verb?L?}\verb?.?{\color{method}\verb?ideal?}\verb?([g1, g2, g3])? \\
	\cIn & {\color{parent}\verb?I?}\verb?.?{\color{method}\verb?groebner_basis?}\verb?()? \\
	\cOut & $(x^2y + y^{-6},$ \\
          & $\ x^3y^{-2} + x^{-6}y,$ \\
          & $\ x^{-2}y + x^{-1}y^{-2},$ \\
          & $\ -xy + x^{-2}y^{-3},$ \\
          & $\ x^2y + x^{-2},$ \\
          & $\ y^{-1} + x^{-1},$ \\
  		  & $\ -y^2 + x^{-1},$ \\
          & $\ x^{-1}y^{-1} + x^{-2}y^{-2})$ \\
\end{tabular}}

\section{Proof of Proposition \ref{prop:multi_div_new}}
\begin{proof}
\label{proof:multi_div_new}
We construct by induction sequences $(f_k)_{k\ge 0}$, $(q_{g,k})_{k \ge 0}$ for $g \in G$ and $(r_k)_{k\ge 0}$ such that for all $k \ge 0$:
\[ f = f_k + \sum_{g \in G}q_{g,k}g + r_k,\]
and $\lt(f_k)_{k\ge 0}$ is strictly decreasing for $ \le_{P} $.

We first set $f_0 = f$, $r_0 = 0$ and $q_{g,0} = 0$ for all $g \in G$. 

If there exists $(i,j) \in L$ and $g \in G$ such that \[\lm \left( \frac{\lm(f_k)}{\lm_{i,j}(g)}g \right) = \lm(f_k) \in T_{i,j} \cap V_{i,<}, \] we set $f_{k+1} = f_k - tg$ and $q_{g,k+1} = q_{g,k} + tg$ where $t = \frac{\lt(f_k)}{\lt_{i,j}(g)}$, and leave unchanged $r_k$ and the other $q_{g,k}$'s. 

Otherwise, we set $f_{k+1} = f_k - lt(f_k)$ and $r_{k+1} = r_k + lt(f_k)$ and leave unchanged the $q_{g,k}$'s. 

By construction, the sequence $(\lt(f_k))_{k \ge 0}$ is strictly decreasing for $ \le_{P} $.
By Remark \ref{rem:convergence_new}, we deduce that 
$\valP(r_{k+1}-r_k)$ and the $\valP(q_{g,k+1}-q_{g,k})$'s tend to $+\infty$ when $k \to +\infty$.
Thus $r_k$ and the $q_{g,k}$'s converge in $K\{\mathbf{X};P\}$.
\end{proof}

\section{Proof of Lemma \ref{lemma:sumSpair_new}}
\begin{proof}
\label{proof:sumSpair_new}

	Write $p_k = \frac{t_kh_k}{c_k}$, $e_k = \sum_{s=1}^{k}c_s$ and $t_k=\gamma_k \tilde{t}_k$ for some $\gamma_k \in K$
	and some monomial $\tilde{t}_k$.
	By hypothesis $u$ is in $T_{i,j}\cap V_{i,<}$ and $u = \tilde{t_k}\lm_{i,j}(h_k) \in T_{i,j}(h_k)\lm_{i,j}(h_k)$ for all $k$. 
	This implies that for all $k <m$ we have:
	\[u \in  T_{i,j}(h_k)\lm_{i,j}(h_k) \cap T_{i,j}(h_{k+1})\lm_{i,j}(h_{k+1}) = T_{i,j}\cdot U((i,j),h_k,h_{k+1})\]
	We deduce that for all $k<m$, there exist $s_k \in T_{i,j}$ and $v_k \in  U((i,j),h_k, h_{k+1}$) such that $u = s_kv_k$.
	Now write
	\begin{align}
	\sum_{k=1}^{m}t_kh_k &= e_1(p_1 - p_2) + \dots + e_{m-1}(p_{m-1}- p_m) + e_mp_m \label{eqn:somme_telescopique_new}
	\end{align}
	For all $k<m$, we have $\lt(t_k h_k)=c_k u= \gamma_k \lc_{i,j}(h_k)\tilde{t}_k \lm_{i,j}(h_k), $
	hence $\frac{t_k}{c_k \tilde{t}_k}=\frac{1}{\lc_{i,j}(h_k)}$.
	For any $k<m$, put $P_k = p_k - p_{k+1}$. We can then write:
	\begin{align*}
		P_k &= \frac{u}{v_k}\left(\frac{v_k}{u}p_k - \frac{v_k}{u}p_{k+1}\right) = \frac{u}{v_k}\left(\frac{t_kv_kh_k}{c_k\tilde{t}_k\lm_{i,j}(h_k)} - \frac{t_{k+1}v_kh_{k+1}}{c_{k+1}\tilde{t}_{k+1}\lm_{i,j}(h_{k+1})} \right)  \\
					  &= \frac{u}{v_k}\left(\frac{1}{\lc_{i,j}(h_k)}\frac{v_k}{\lm_{i,j}(h_k)}h_k - \frac{1}{\lc_{i,j}(h_{k+1})}\frac{v_k}{\lm_{i,j}(h_{k+1})}h_{k+1}\right)	\\
					  &= \frac{1}{\lc_{i,j}(h_k)\lc_{i,j}(h_{k+1})}\frac{u}{v_k}\left(\lc_{i,j}(h_{k+1})\frac{v_k}{\lm_{i,j}(h_k)}h_k - \lc_{i,j}(h_k)\frac{v_k}{\lm_{i,j}(h_{k+1})}h_{k+1}\right) \\
					  &= \frac{1}{\lc_{i,j}(h_k)\lc_{i,j}(h_{k+1})}\frac{u}{v_k}S((i,j),h_k,h_{k+1}).
	\end{align*}
	Plugging in the last expression back into equation \eqref{eqn:somme_telescopique_new} gives the desired equality
	with $d_k = \frac{e_k}{\lc_{i,j}(h_k)\lc_{i,j}(h_{k+1})}$ and $t^{\prime}_m = \frac{e_m}{c_m}t_m$.
	It satisfies \ref{enum:1_new}.
	The hypothesis forces $\val{}(e_m) > \val{}(c_m)$. Then we have $\valP(t^\prime_mh_m) = \val{}(e_m) + \valP(u) > \val{}(c_m) + \valP(u) = \valP(c_1u)$, which proves \ref{enum:2}.
	In addition, $\frac{u}{v_k} = s_k \in T_{i,j}$, which proves \ref{enum:3_new}.
	Finally, using that $\val (e_k) \geq c$ and $d_k = \frac{e_k}{\lc_{i,j}(h_k)\lc_{i,j}(h_{k+1})}$, one gets \ref{enum:4_new}.
\end{proof}

\section{Proof of Proposition \ref{prop:buch_criterion_new}}
\label{app:proof_criterion}
Before presenting the proof, we establish two lemmas that are adaptations of 
Lemma \ref{lemma:Spair} and Lemma \ref{lemma:inequality}, respectively.

\begin{lemma}
\label{lemma:Spair_new}
For $f,g \in \LaurentPolytopalRing$, $(i,j) \in L$ and $v \in \lm_{i,j}(f)T_{i,j}(f) \cap \lm_{i,j}(g)T_{i,j}(g)$, we have
$\lt(S((i,j),f,g,v)) <_{P} \lc_{i,j}(f)\lc_{i,j}(g)v$.
\end{lemma}

\begin{proof}
Since $v \in T_{i,j}(f)\lm_{i,j}(f)\cap T_{i,j}(g)\lm_{i,j}(g)$, there exists $m_f \in T_{i,j}(f)$
and $m_g \in T_{i,j}(g)$ such that $v = \lm(m_ff) = \lm(m_gg) = m_f\lm_{i,j}(f) = m_g\lm_{i,j}(g)$. 
Then the leading terms of $\lc_{i,j}(g)\frac{v}{\lm_{i,j}(f)}f$ and  $\lc_{i,j}(f)\frac{v}{\lm_{i,j}(g)}g$
are both equal to $\lc_{i,j}(f)\lc_{i,j}(g)v$.
They cancel out leaving $\lt(S((i,j),f,g,v)) <_{P} \lc_{i,j}(f)\lc_{i,j}(g)v$.
\end{proof}

\begin{lemma}
\label{lemma:inequality_new}
If $f \in \LaurentPolytopalRing$ and $(i,j) \in L$ are such that $\lt (f) <_{P} u$ for some term $ u $ 
satisfying $\lm(u) \in T_{i,j}\cap V_{i,<}$,
then for any $v \in T_{i,j}$ we have $ \lt(vf) <_P vu.$
\end{lemma}

\begin{proof}
The coefficient of $ u $ play no role in the proof, so we can
assume that $ u = \lm(u) $.
Take $t$ a term of $f$. Then $t <_{P} u$. Since $ u, v \in V_i$,
we have by Lemma \ref{lem:tij_module} that $ \val_{P}(uv) = \val_{r_i}(u) + \val_{r_i}(v) = \val_P(u) + \val_P(v) $.
We separate in 3 cases following Definition \ref{def:new_order}:
\begin{itemize}
\item Case $ \val_P(u) < \val_P(t) $.

Then $ \val_P(uv) = \val_P(u) + \val_P(v) < \val_P(t)
+ \val_P(v) \le \val_P(tv)$. Thus $ tv <_{P} uv $.

\item Case $ \val_P(u) = \val_P(t) $ and $ \min(I_P(u)) < \min(I_P(t))$.

Then $ \val_P(uv) = \val_P(u) + \val_P(v) = \val_P(t) + \val_P(v) \le \val_P(tv) $.
If $ \val_P(uv) < \val_P(tv)$, then $ tv <_p uv$ and we are done, so let's suppose
$ \val_P(uv) = \val_P(tv) $. Then we have
\begin{equation}
	\label{eq:item_2}
	\val_P(uv) = \val_{r_i}(u) + \val_{r_i}(v) = \val_P(tv) = \min_{k \in I_P}(\val_{r_k}(t) + \val_{r_k}(v))
\end{equation}

We have $ u \in V_{i,<} $, so $ i = \min(I_P(u))$.
Now since $ \val_P(u) = \val_P(t) $ and $ i = \min(I_P(u)) < \min(I_P(t))$,
we have $ \val_{r_k}(t) > \val_{r_i}(u)$ for $ k \le i $.
Also, since $ v \in T_{i,j} \subseteq V_{i}$, we have
$ \val_{r_k}(v) \ge \val_{r_i}(v) $ for any $ k \in I_P $.
Thus for $ k \le i $, we have
$ \val_{r_k}(t) + \val_{r_k}(v) > \val_{r_i}(u) + \val_{r_i}(v) $. It follows that
the minimum in (\ref{eq:item_2})
can be reached only for a $ k > i $. This shows $ \min(I_P(tv)) > \min(I_P(uv)) $ and so $ tv <_P vu $.

\item Case $ \val_P(u) = \val_P(t) $ and $ \min(I_P(u)) = \min(I_P(t))$ and $ u >_{\omega} t $

Reasoning like in case 2 with $ \min(I_P(u)) = \min(I_P(t)) $ this time, one
gets that the minimum in (\ref{eq:item_2}) can only be reached
for indices $ k \ge i $, so $ \min(I_P(tv)) \ge \min(I_P(uv)) $.
If the inequality is strict, we get $ tv <_P uv$ and we are done, otherwise
this is $ \le_{\omega} $ who breaks ties between $ uv $ and $ tv $.
We can then just use item 2. of Definition \ref{def:gmo} like in the proof of Lemma \ref{lemma:inequality} to get
$ tv <_{\omega} uv $, and so $ tv <_P uv $.

\end{itemize}

Now since the maximum on the $vt$'s for $ \le_P $ equals $ \lt(vf) $, we conclude
$ \lt(vf) <_P vu $.
\end{proof}

\begin{proof}[Proof of Proposition \ref{prop:buch_criterion_new} ]
\label{proof:buch_criterion_new}
By contradiction, assume that (2) is true and that $H$ is not a Gröbner basis of $J$. 
Then there exists $f \in  J$ such that $\lm(f) \notin \bigcup_{(i,j) \in L, 1\le k \le m}T_{i,j}(h_k)\lm_{i,j}(h_k)$.
Since, $f \in J = (h_1,\dots,h_m)$, we can write $f = \sum_{k=1}^{m}q_kh_k$ for some $q_k$ in $\LaurentPolytopalRing$. 
%By summing for each $j$ over the terms of $q_k$, we can rewrite $f$ as $\sum_{j=1}^{m}\sum_{\alpha \in \Delta(j)}t_{j,\alpha}h_k$ where for each $j$, $(t_{j,\alpha})_{\alpha \in \Delta(j)}$ is a finite family of terms. 
Write $\Delta(k)$ to be the set of terms of $q_k$.
We can rewrite $f$ as $\sum_{k=1}^{m}\sum_{\alpha \in \Delta(k)}t_{k,\alpha}h_k$.
For such a writing of $f$, define $u = \max\{\lt(t_{k,\alpha}h_k), 1 \le k \le m, \alpha \in \Delta(k)\}$
and write the term $u$ as $u=c\tilde{u}$
for some $c \in K$ and some monomial $\tilde{u}.$
We have $\lt(f) <_P u$ because $\lm(f) \notin \cup_{i,j}T_{i,j}(h_k)\lm_{i,j}(h_k)$.
Thus, $\val_P(u)$ is upper-bounded.

Since $\val$ is discrete, there is a maximal $\val_P(u)$ among all
possible expressions of $f=\sum_{k=1}^{m}q_kh_k$.
Among the expressions reaching this valuation, Lemma \ref{lemma:stric_descending_seq} ensure 
there is one such that $u$ is minimal. Let $(i,j) \in L$ be such that $u \in T_{i,j} \cap V_{i,<}$.
Define $Z = \{(k,\alpha) \in \llbracket 1,m\rrbracket \times \Delta(k),\ \textrm{s.t. }
\lt(t_{k,\alpha}h_k) =_P u\}$ and
$Z^\prime = \{(k, \alpha) \in \llbracket 1,m\rrbracket \times \Delta(k),\ \textrm{s.t. } \lt(t_{k,\alpha}h_k) <_P u\}$. 
We can then write:
\begin{equation}
 f = \sum_{(k,\alpha) \in Z}t_{k,\alpha}h_k + \sum_{(k,\alpha) \in Z^\prime}t_{k,\alpha}h_k \label{eq:ecriture_f_ij_new}
\end{equation}
Let $g := \sum_{(k,\alpha) \in Z}t_{k,\alpha}h_k$.
We have $\lt(g) \le_P \max(\lt(f),\lt(\sum_{(k,\alpha) \in Z^\prime}t_{k,\alpha}h_k )) <_{P} u$ and
$\lt(t_{k,\alpha}h_k) = c_{k,\alpha}\tilde{u}$ for all $(k,\alpha) \in Z$, where the $c_{k,\alpha}$ all have the same valuation.
So $g$ satisfies the conditions of Lemma \ref{lemma:sumSpair_new} and we can write
\begin{equation}
	g = \sum_{k=1}^{m-1}d_k\frac{\tilde{u}}{v_k}S((i,j),h_k,h_{k+1},v_k) + t^\prime_mh_m \label{eq:def_g_ijn_Buchberger}
\end{equation}
for some $d_k \in K$, $v_k \in U((i,j),h_k,h_{k+1})$, $\val \left( d_k \lc_{i,j}(h_k) \lc_{i,j}(h_{k+1}) \right) \geq \val (c)$
and $\tilde{u}/v_k \in T_{i,j}$ for $k<m$, and with $\lt(t^\prime_mh_m) <_P u$.
Now we use the hypothesis that all the S-pairs of elements of $H$ reduce to zero. 
For each $k<m$ we can write
\begin{equation}
S((i,j),h_k,h_{k+1},v_k) = \sum_{l=1}^{m} q^{(k)}_l  h_l,
\end{equation}
for some $q^{(k)}_l $'s in $\LaurentPolytopalRing$
satisfying 
\begin{align*}
\lt (q^{(k)}_l  h_l) & \leq_P \lt \left( S((i,j),h_k,h_{k+1},v_k) \right), \\
                    & <_P \lc_{i,j}(h_k)\lc_{i,j}(h_{k+1})v_k,
\end{align*}
where the last inequality comes from Lemma \ref{lemma:Spair_new}.
Since $v_k \in T_{i,j} \cap V_{i,<}$ and $\tilde{u}/v_k \in T_{i,j}$, we can apply
Lemma \ref{lemma:inequality_new}:

\begin{align*}
  \lt \left(\frac{\tilde{u}}{v_k}q^{(k)}_l  h_l\right) & <_P \lc_{i,j}(h_k)\lc_{i,j}(h_{k+1})v_k\frac{\tilde{u}}{v_k},  \\
 & = \lc_{i,j}(h_k)\lc_{i,j}(h_{k+1})\tilde{u}.
\end{align*}
Finally, using that $\val \left( d_k \lc_{i,j}(h_k) \lc_{i,j}(h_{k+1}) \right) \geq \val (c)$,
we deduce that for all $l \in \llbracket 1, m \rrbracket$ and $k \in \llbracket1,m-1 \rrbracket$; \[\lt \left(d_k \frac{\tilde{u}}{v_k}q^{(k)}_l  h_l\right)<_P u.\]
Inserting the expressions of $d_k \frac{\tilde{u}}{v_k}S((i,j),h_k,h_{k+1},v_k)$ as $\sum_{l=1}^{m} d_k \frac{\tilde{u}}{v_k}q^{(k)}_l  h_l$
in Equations \eqref{eq:def_g_ijn_Buchberger} and then \eqref{eq:ecriture_f_ij_new},
we get an expression of $f$ in terms of the $h_k$'s with strictly smaller $u$ for $ \le_P $, contradicting its minimality.
\end{proof}

\end{appendices}

\end{document}